\documentclass[10pt]{article}
\usepackage{latexsym}
\usepackage{amssymb}
\usepackage{amsthm}
\usepackage{amsmath}
\usepackage{graphicx}
\usepackage[latin1]{inputenc}
\usepackage{hyperref}
\newtheorem{definition}{Definition}[section]
\newtheorem{theorem}{Theorem}[section]
\newtheorem{prop}[theorem]{Proposition}
\newtheorem{coro}[theorem]{Corollary}
\newtheorem{lemma}[theorem]{Lemma}
\newtheorem{remark}[theorem]{Remark}

\newcommand{\R}{\mathbb{R}}             
\newcommand{\N}{\mathbb{N}}             
\newcommand{\Z}{\mathbb{Z}}             %
\newcommand{\C}{\mathbb{C}}             
\renewcommand{\S}{\mathbb{S}}             

\author{Damien Gobin\footnote{Department of Mathematics and Statistics, McGill University, 805 Sherbrooke South West, Montr\'eal QC, H3A 2K6.
Email adress: damien.gobin@mcgill.ca.}}
\title{Inverse scattering at fixed energy for radial magnetic Schr\"odinger operators with obstacle in dimension two}
\date{\today}

\begin{document}

\maketitle


\begin{abstract}
We study an inverse scattering problem at fixed energy for radial magnetic Schr\"odinger operators on $\R^2 \setminus B(0,r_0)$, where $r_0$ is a positive and
arbitrarily small radius. We assume that the
magnetic potential $A$ satisfies a gauge condition and we consider 
the class $\mathcal{C}$ of smooth, radial and compactly supported electric potentials and magnetic fields denoted by $V$ and $B$ respectively. If $(V,B)$ and $(\tilde{V},\tilde{B})$
are two couples belonging to $\mathcal{C}$, we then show that if the corresponding phase shifts $\delta_l$ and $\tilde{\delta}_l$ (i.e. the scattering data at fixed energy)
coincide for all $l \in \mathcal{L}$, where $\mathcal{L} \subset \N^{\star}$ satisfies the M\"untz condition $\sum_{l \in \mathcal{L}} \frac{1}{l} = + \infty$, then
$V(x) = \tilde{V}(x)$ and $B(x) = \tilde{B}(x)$ outside the obstacle $B(0,r_0)$. The proof use the Complex Angular Momentum method and is close in spirit to the celebrated
B\"org-Marchenko uniqueness Theorem.

\vspace{0.5cm}
\noindent \textit{Keywords}. Inverse Scattering, radial magnetic Schr\"odinger operators, phase shifts.\\
\textit{2010 Mathematics Subject Classification}. Primaries 81U40, 35P25; Secondary 58J50.
\end{abstract}

\tableofcontents

\section{Introduction and statement of the main result}

In this work we are interested in an inverse scattering problem at fixed energy for magnetic Schr\"odinger operators on $\R^2 \setminus B(0,r_0)$ where $B(0,r_0)$ is
an obstacle whose radius $r_0$ is positive and arbitrarily small.
In quantum scattering theory, we study a pair of Hamiltonians $(-\Delta,H_{A,V})$ on $L^2(\R^2)$, where $H_{A,V}$ denotes the quantum magnetic Schr\"odinger Hamiltonian which 
describes the interaction of a charged particle with an electric field $\nabla V$ and a magnetic field $B=dA$. The object of main interest is then the scattering operator
$S$. Since the scattering operator $S$ commutes with the free operator $-\Delta$ it can be reduced to a multiplication by an operator-function $S(\lambda)$, called 
the scattering matrix at energy $\lambda$, in the spectral representation of the Hamiltonian $-\Delta$. The question we usually adress in an inverse scattering problem at fixed energy
is the following:
\begin{center}
\emph{Does the scattering matrix $S(\lambda)$ at a fixed energy $\lambda > 0$ uniquely determine the electric potential $V$ and the magnetic field $B$ outside the obstacle?}
\end{center}
The aim of this paper is then to prove that, up to a gauge choice, we can answer positively to this question for radial, smooth and compactly supported electric potentials and magnetic fields.
Roughly speaking, we will thus show that, if the magnetic potential $A$ satisfies a particular gauge condition, for such electric potentials and
magnetic fields, if $\lambda > 0$ is a fixed energy:
\[ S(\lambda) = \tilde{S}(\lambda) \quad \Rightarrow \quad V(x) = \tilde{V}(x) \quad \text{and} \quad B(x) = \tilde{B}(x), \quad \forall x \in \R^2 \setminus B(0,r_0).\]
Inverse scattering at fixed energy for (magnetic) Schr\"odinger operators has been largely studied since the end of the twentieth century and is a tricky question.
For instance, concerning Schr\"odinger operators (with no magnetic field), even if for exponentially decreasing potentials (see Novikov's papers
\cite{Nov1,Nov2}) or for particular classes of radial potentials (see Daud\'e and Nicoleau's paper \cite{DN5}), we can answer positively
to this question, we emphasize that, in general, the answer is negative. Indeed, in dimension two, Grinevich and Novikov
construct in \cite{GN} a family of real spherically symmetric potentials in the Schwartz space such that the associated scattering matrices 
are equal to the identity (transparent potentials). We also mention the work \cite{Sab} of Sabatier where a class of radial transparent potentials is obtained
in the three-dimensional case.

The main tool of this paper consists in complexifying the angular momentum that appears in the reduction of the Schr\"odinger operator into a countable family of 
one-dimensional radial equations, i.e. in the separation of variables procedure.
Indeed, thanks to variables separation, the scattering matrix at fixed energy can be decomposed onto a family of scattering coefficients $\delta_l$, called the phase
shifts, indexed by a discrete set of angular momentum $l \in \Z$.
The approach consisting in complexifying the angular momentum is called the Complex Angular Momentum (CAM) method. The idea of this method is the following.
We first allow the angular momentum
$l \in \Z$ to be a complex number $\nu \in \C$. In some cases it is then possible to extend the equality $\delta_l = \tilde{\delta}_l$ for all $l \in \Z$ into the 
equality $\delta(\nu) = \tilde{\delta}(\nu)$ for all $\nu \in \C \setminus \{\text{poles}\}$. Indeed, functions in some particular classes of holomorphic functions, are
completely determined by its values on a sufficiently large subset of the integers (Nevanlinna's class). We then use this new amount of
informations to get the equality between the electric potentials and the magnetic fields.
The general idea of considering complex angular momentum originates from a paper of Regge 
(see \cite{Reg}) as a tool in the analysis of the scattering matrix of Schr\"odinger operators in $\R^3$ with spherically symmetric potentials.
We also refer to \cite{AR,New} for books dealing with this method. We mention that this tool was already used in the field of inverse problems for one angular momentum in
\cite{DN5,Lo,Ram} for Schr\"odinger operators,
in \cite{DN3,DN} in the context of general relativity, in \cite{Pap1,DKN} on asymptotically hyperbolic manifolds and in \cite{DNK2} to study counterexamples 
for the Calder\'on problem which is closely related to inverse scattering problems at fixed energy on asymptotically hyperbolic manifolds.
Moreover, this method was also used for two angular momenta in \cite{G2} and we note that it is also a useful tool in high energy physics (see \cite{Collins}).
This work is a continuation and is really close to the spirit of the paper \cite{DN5} of Daud\'e and Nicoleau in which the authors treat the same question for the Schr\"odinger
operators with no magnetic fields in all dimensions and for particular classes of radial potentials.

\subsection{Description of our framework}
\noindent
In this work, we study an inverse scattering problem for magnetic Schr\"odinger operators on the region $\Omega = \R^2 \setminus B(0,r_0)$, where $r_0 > 0$ is fixed.
We thus study the Hamiltonian,
\[  H = (D-A)^2 + V = -\Delta + 2i A . \nabla + A^2 + i \mathrm{div}(A) + V.\]
where $D = \frac{1}{i} \nabla$, $V$ is the electric potential and $A = A_1 dx^1 + A_2 dx^2$ is the $1$-form corresponding to the magnetic potential.
The magnetic field $B$ is the $2$-form defined as $B=dA$ which can be identified with the antisymmetric $2 \times 2$
matrix
\[ B = dA = \begin{pmatrix}
   0 & b \\
   -b & 0
\end{pmatrix}, \quad b_{j,k} = \partial_{x_j}A_k - \partial_{x_k}A_j.\]
We then need to add a transversal gauge condition on the magnetic potential $A$ given by the following Definition.

\begin{definition}[Gauge choice]\label{defclassA}
 We say that a magnetic potential $A$ belongs to the class $\mathcal{A}$ if it is smooth on $\R^2$ and satisfies the gauge condition
 \[A(x) = - \int_0^1 sB(sx) . x \, ds.\]
\end{definition}
\noindent
From now on, we assume that $A \in \mathcal{A}$. Thanks to this gauge condition we obtain the following Lemma.

\begin{lemma}\label{defAgamma}
 Assume that $A \in \mathcal{A}$, then
 \begin{enumerate}
  \item (Transversality condition)
  \[A(x) . x = 0.\]
  \item $A$ satisfies
  \begin{equation}\label{lienAgamma}
   A(x) = \frac{\gamma(r,\theta)}{r} (-\sin(\theta),\cos(\theta)),   
  \end{equation}
 where
 \begin{equation}\label{defgamma1}
 \gamma(r,\theta) = \int_0^r b(\tau \cos(\theta), \tau \sin(\theta)) \tau \, d\tau.  
 \end{equation}
 \end{enumerate}
\end{lemma}
\noindent
We now summarize the assumptions we need on the electric potential $V$ and the magnetic field $B$ in the following Definition.

\begin{definition}\label{defVB}
 We say that a couple $(V,B)$ where $V$ is an electric potential and $B$ is a magnetic field associated with a magnetic potential $A \in \mathcal{A}$ belongs to the class
 $\mathcal{C}$ if:
 \begin{enumerate}
  \item $V$ and $B$ are both radial.
  \item $V$ and $B$ are both compactly supported.
  \item $V$ is piecewise continuous and $B$ is smooth.
 \end{enumerate}
\end{definition}
\noindent
From now on, we will always assume that $(V,B)$ is a couple of $\mathcal{C}$. For such a couple we can prove the following Lemma.

\begin{lemma}
Assume that $(V,B) \in \mathcal{C}$, then:
\begin{enumerate}
 \item $\gamma(r,\theta) = \gamma(r)$.
 \item $\mathrm{div}(A) = 0$.
 \item $A^2 = \frac{\gamma(r)^2}{r^2}$ is a radial function.
 \item $A(x) . \nabla = \frac{\gamma(r)}{r^2} D_{\theta}$, where $D_{\theta} = -i \partial_{\theta}$.
\end{enumerate} 
\end{lemma}
\noindent
Thanks to these properties we obtain that the Hamiltonian $H$ can be rewritten as
\begin{equation}\label{hamcplt}
  H = - \Delta - \frac{2\gamma(r)}{r^2} D_{\theta} + \frac{\gamma(r)^2}{r^2} + V(r).
\end{equation}
We then know (see \cite{DN5,RS3}) that the Hamiltonian $H$ can be reduced into a countable family of radial Hamiltonians.
Indeed, let us introduce the decomposition
\[ L^2(\R^2) = L^2(\R^+,r dr) \otimes L^2(\mathbb{S}^1,d\sigma).\]
For functions of the form $u(x) = f(r)g(\theta)$, where $r = |x| > 0$ and $\theta = \frac{x}{r} \in \S^1$, we then obtain
\begin{equation}\label{hamdecang}
 Hf(r)g(\theta) =  \left( - \frac{d^2}{dr^2} - \frac{1}{r} \frac{d}{dr} + \frac{\gamma(r)^2}{r^2} + V(r) + \Delta_{\S^1} - \frac{2\gamma(r)}{r^2} D_{\theta} \right)f(r)g(\theta),
\end{equation}
where the operator $\Delta_{\S^1}$ is the Laplace-Beltrami operator on the circle $\S^1$. Its eigenvalues are given by $-l^2$, for $l \in \Z$.
It then follows that
\[  L^2(\R^+,r dr) \otimes L^2(\mathbb{S}^1,d\sigma) = \underset{l \in \Z}{\oplus} L^2(\R^+,r dr) \otimes K_l,\]
where $K_l$ is the invariant subspace of the eigenspace of $\Delta_{\S^1}$, associated with the eigenvalue $-l^2$ for $l \in \Z$, generated by $e^{il\theta}$.
Therefore, by introducing
\[L_l = L^2(\R^+,r dr) \otimes e^{il \theta}, \quad l \in \Z,\]
we thus obtain that the restriction of the Hamiltonian $H$ on each subspace $L_l$ is given by
\[H_{| L_l} = -\frac{d^2}{dr^2} - \frac{1}{r} \frac{d}{dr} + \frac{l^2}{r^2} - \frac{2l \gamma(r)}{r^2} + \frac{\gamma(r)^2}{r^2} + V(r).\]
Finally, we can get rid of the first order term by defining the unitary operator
\[\begin{array}{ccccc}
U & : & L^2(\R^+,r dr) & \to & L^2(\R^+,dr) \\
 & & f & \mapsto & r^{\frac{1}{2}} f(r) \\
\end{array}.\]
Indeed, conjugating $H_{| L_l}$ by $U$, we then obtain a new family of radial Hamiltonians $UH_{| L_l} U^{-1}$, denoted by $H_{l}$, given by
\begin{equation}\label{hampart}
  H_l := -\frac{d^2}{dr^2} + \frac{l^2 - \frac{1}{4}}{r^2} - \frac{2l \gamma(r)}{r^2} + \frac{\gamma(r)^2}{r^2} + V(r).
\end{equation}

\begin{remark}
 In comparison with the non-magnetic case studied in \cite{DN5} we note that the new term in the potential is $- \frac{2l \gamma(r)}{r^2}+ \frac{\gamma(r)^2}{r^2}$.
 The main fact to note is the presence of $l$ in the potential that makes it necessary to obtain better estimates in terms of $l$ on the Green kernels
 than the ones obtained in \cite{DN5}.
\end{remark}
\noindent
Since we assumed that the magnetic field $B$ and the electric potential $V$ have compact supports included in a ball $B(0,R)$, $R > 0$, we note (see (\ref{defgamma1})) that
\begin{equation}\label{gammargd}
 \gamma(r) = \gamma(R) = \frac{\mathrm{Flux}(B)}{2\pi}, \quad \forall r \geq R,
\end{equation}
where $\mathrm{Flux}(B)$ denotes the magnetic flux of $B$. Therefore,
\[ H_l = -\frac{d^2}{dr^2} + \frac{l^2 - \frac{1}{4}}{r^2} - \frac{2l \gamma(R)}{r^2} + \frac{\gamma(R)^2}{r^2} = -\frac{d^2}{dr^2} + \frac{(l-\gamma(R))^2 - \frac{1}{4}}{r^2}, \quad \forall r \geq R.\]

\begin{remark}
 We note that even if $R \leq r_0$, i.e. in particular that the support of the magnetic field $B$ is included in the obstacle $B(0,r_0)$, it still has an influence outside
 the obstacle since the magnetic flux appears in the corresponding Hamiltonian. This is a consequence of the Aharonov-Bohm effect. In this case we can only hope to 
 reconstruct the flux of the magnetic field.
\end{remark}
\noindent
Since we work on $\Omega = \R^2 \setminus B(0,r_0)$, the previous equality tells us that the reference operator we should use to study the dynamic of $H_l$
is given by
\[ H_l^0 = -\frac{d^2}{dr^2} + \frac{(l-\gamma(R))^2 - \frac{1}{4}}{r^2}, \quad \forall r \geq r_0.\]
For the sake of clarity we introduce two arguments
\[ \nu  :=l \quad \text{and} \quad \nu_R := \nu - \gamma(R) = l-\gamma(R).\]
We can then rewrite $H_{\nu}$ in terms of $H_{\nu}^0$ by the following way
\begin{equation}\label{complibre}
  H_{\nu} = H_{\nu}^0 + q_{\nu}(r)
\end{equation}
where the potential
\[ q_{\nu}(r) := - \frac{2\nu (\gamma(r)-\gamma(R))}{r^2} + \frac{\gamma(r)^2-\gamma(R)^2}{r^2} + V(r),\]
vanishes for all $r \geq R$.

\begin{remark}
 Here is a reason why it is easier to work on $\Omega$ instead of $\R^2$. Indeed, the potential $q_{\nu}$ has a quadratic singularity as $r$ tends to $0$ and thus doesnot
 satisfy the hypothesis $(H_1)$ introduced in \cite{DN5}. It is this singularity that prevent us from obtain the corresponding inverse scattering result on the whole 
 plane $\R^2$.
\end{remark}

\subsection{Statement of the main result}
\noindent
The aim of this section is to introduce the scattering data we will study in the following in order to state our main result.
We refer the reader to \cite{DN5,Lo} for more details.
First of all we recall that we study the family of Hamiltonians $H_{\nu}$ given by (\ref{hampart}) on $\Omega = \R^2 \setminus B(0,r_0)$ where $r_0 > 0$ is fixed.
\noindent
Following Regge's idea, we consider the radial Schr\"odinger equation on $(r_0,+\infty)$ at the fixed energy $\lambda = 1$, where the angular
momentum $\nu = l$ is now assumed to be a complex number,
\begin{equation}\label{eqstat}
 -u'' + \left( \frac{(\nu_R)^2 -  \frac{1}{4}}{r^2} + q_{\nu}(r) \right) u = u. 
\end{equation}
We note that when $\nu_R = \nu - \gamma(R) = l - \gamma(R)$ we recover the family of Hamiltonians (\ref{hampart}) obtained previously in the separation of variables procedure. First,
since the potential $q_{\nu}$ is compactly supported, we can define the Jost solutions $F^{\pm}(r,\nu)$ as the unique solutions of (\ref{eqstat}) satisfying the boundary
conditions at $r = + \infty$
\begin{equation}\label{asJost}
 F^{\pm}(r,\nu) \sim e^{\pm i r}, \quad r \to + \infty. 
\end{equation}
\noindent
For every fixed $r \geq r_0$, the maps $\nu \mapsto F^{\pm}(r,\nu)$ and $\nu \mapsto F^{\pm'} (r,\nu)$ are holomorphic on the whole complex plane $\C$.
Moreover,
\[ \overline{F^{\pm}(r,\nu)} = F^{\mp}(r,\overline{\nu}), \quad \forall \nu \in \C,\]
and we note that the pair of Jost solutions is a fondamental system of solutions of (\ref{eqstat}).
\noindent
Secondly, we define the regular solution, denoted by $\Phi(r,\nu)$, as the solution of (\ref{eqstat}) satisfying the Dirichlet condition at $r = r_0$:
\[ \Phi(r,\nu)_{|r=r_0} = 0.\]
Since the pair of Jost solutions is a fondamental system of solutions of (\ref{eqstat}) there exists two functions of $\nu$, $\alpha$ and $\beta$, such that
\begin{equation}\label{defsolreg}
 \Phi(r,\nu) = \alpha(\nu) F^+(r,\nu) + \beta(\nu) F^-(r,\nu).
\end{equation}
Moreover, by uniqueness, $\Phi(r,\nu)$ is in fact given, up to a multiplicative constant, by
\[ \Phi(r,\nu) = i\left( F^-(r_0,\nu) F^+(r,\nu) - F^+(r_0,\nu) F^-(r,\nu) \right).\]
We thus note that for every fixed $r \geq r_0$ the maps $\nu \mapsto \Phi(r,\nu)$ and $\nu \mapsto \Phi'(r,\nu)$ are holomorphic on the complex plane $\C$ and moreover
\[ \overline{\Phi(r,\nu)} = \Phi(r,\overline{\nu}), \quad \forall \nu \in \C.\]
The functions 
\[ \alpha(\nu) = iF^-(r_0,\nu)\]
and
\[ \beta(\nu) = -i F^+(r_0,\nu)\]
are then called the Jost functions. It follows from (\ref{asJost}) that
\[ W(F^+(r,\nu),F^-(r,\nu)) = -2i,\]
where the Wronskian of two functions $u$ and $v$ is given by $W(u,v) = uv' - u'v$. Hence,
\[ \alpha(\nu) = \frac{i}{2} W(\Phi(r,\nu),F^-(r,\nu)),\]
and
\[ \beta(\nu) = -\frac{i}{2} W(\Phi(r,\nu),F^+(r,\nu)).\]
From the definition of the Jost functions we immediately deduce that these functions are holomorphic on the complex plane $\C$ and satisfy
\begin{equation}\label{lienab}
 \overline{\alpha(\nu)} = \beta(\overline{\nu}).
\end{equation}
We can now introduce the so-called Regge interpolation function
\[ \sigma(\nu) = e^{i \pi \left( \nu + \frac{1}{2} \right) } \frac{\alpha(\nu)}{\beta(\nu)}.\]
We note that, when $\nu \in \R$, it follows from (\ref{lienab}) that $|\sigma(\nu)| = 1$. We can thus define the generalized phase shifts $\delta(\nu)$ as a continuous 
function on $\R$ through the relation
\begin{equation}\label{defdelta1}
 \sigma(\nu) = e^{2i\delta(\nu)}. 
\end{equation}
We emphasize that the quantities $\delta(\nu)$ for $\nu = l$, where $l \in \Z$, are related to the physical phase shifts and are quantities that we can measure by a scattering
experiment.
\noindent
The question we adress is then the following:
\begin{center}
\emph{Does the phase shifts $\delta_l$ uniquely determine the electric potential $V$ and the magnetic field $B$?}
\end{center}
The aim of this paper is thus to prove the following Theorem.

\begin{theorem}\label{mainthm}
 Let $(V,B)$ and $(\tilde{V},\tilde{B})$ be two couples belonging to the class $\mathcal{C}$ (see Definition \ref{defVB}). Let $\mathcal{L} \subset \N^{\star}$ be
 a subset satisfying the M\"untz condition
 \[ \sum_{l \in \mathcal{L}} \frac{1}{l} = + \infty.\]
 Assume that
 \[ \delta_l = \tilde{\delta}_l, \quad \forall l \in \mathcal{L}.\]
 Then
 \[ V(x) = \tilde{V}(x) \quad \text{and} \quad B(x) = \tilde{B}(x), \quad \forall |x| \geq r_0.\]
\end{theorem}

\begin{remark}
 \begin{enumerate}
  \item We emphasize that in the proof of Theorem \ref{mainthm} we first reconstruct the function $\gamma(r)$ for all $r \geq r_0$ and we can thus reconstruct the 
  magnetic potential $A$. It is not surprising that we can actually reconstruct the magnetic potential $A$ since
  we assumed that it satisfies a gauge condition given by Definition \ref{defclassA}.
  \item If we assume that the support of $B$ is included in the ball $B(0,r_0)$ we reconstruct the electric potential and the magnetic flux and we cannot hope to 
  obtain more informations.
  We emphasize that this is a consequence of the Aharonov-Bohm effect which roughly speaking tells us that the magnetic field has an influence outside the ball $B(0,r_0)$ 
  even if its support is included in $B(0,r_0)$.
  \item As mentionned previously we really need to introduce a parameter $r_0$ to avoid the neighborhood of $r = 0$. Indeed, our potential 
  \[ q_{\nu}(r) = - \frac{2\nu (\gamma(r)-\gamma(R))}{r^2} + \frac{\gamma(r)^2-\gamma(R)^2}{r^2} + V(r),\]
  has a quadratic singularity as $r$ tends to $0$ and thus doesnot satisfy the hypothesis $(H_1)$ introduced in \cite{DN5}. It is this singularity that prevent us from
  obtain the corresponding inverse scattering result on $\R^2$. Actually, we can note that the free dynamics on the whole half line $(0,+\infty)$ have to be
  \[ H^{0,+} = -\frac{d^2}{dr^2} + \frac{(l-\gamma(R))^2 - \frac{1}{4}}{r^2}, \quad \text{as} \quad r \to +\infty,\]
  and
  \[ H^{0,-} = -\frac{d^2}{dr^2} + \frac{l^2 - \frac{1}{4}}{r^2}, \quad \text{as} \quad r \to 0,\]
  and we have to collapse these two dynamics to be able to compare our objects with the ``free'' ones. We still work on this question.
  \item We really need to work with an electric potential $V$ and a magnetic field $B$ compactly supported. Indeed, it allows us to use uniform asymptotics of Bessel 
  functions and then to obtain estimates on the Green kernels sufficiently good in terms of $\nu$ to conclude despite the 
  presence of $\nu$ in the potential $q_{\nu}$ (see Appendix \ref{appe}).
  \item In \cite{DN5} the authors obtain a local result in nature since they prove (for particular classes of electric potentials) that if the phase shifts are super-exponentially
  close, that is to say $\delta_l - \tilde{\delta}_l = O(e^{-Al})$ for all $A > 0$, then the corresponding electric potentials coincide. To do that the authors use a 
  uniqueness result for functions in the Hardy class. We expect that the same result is also true in our framework and we emphasize that we just need to adapt the proof of Section
  \ref{proofp1} to obtain it.
 \end{enumerate}
\end{remark}

\noindent
As mentionned previously direct and inverse scattering for magnetic Schr\"odinger operators are subjects of main interest. Let us give some important references in these
fields.
\begin{itemize}
 \item {\it Direct scattering for magnetic Schr\"odinger operators.} First, let us recall that the Aharonov-Bohm effect was introduced in \cite{AB} and was for instance
 also studied in \cite{Rui}. The direct scattering for magnetic Schr\"odinger operators was for instance the object of the papers \cite{Nic1,dOP,PR,RY,Rui,Tam,Y}.
 \item {\it Inverse scattering for magnetic Schr\"odinger operators in dimension $n \geq 3$.} Let us first recall some inverse scattering results from the knowledge of the 
 scattering matrix for more than one fixed energy. In \cite{ER} the authors reconstruct the electric potential and the magnetic field from the knowledge of the 
 scattering matrix at high energies if the magnetic and electric potentials are exponentially decreasing. In \cite{Nic2} the authors solve the inverse scattering
 problem from the knowledge of the scattering matrix at all energies for electric potentials with short-range and compactly supported magnetic fields. There also exist some 
 inverse scattering results at fixed energy. In \cite{ER2} the authors solve the inverse scattering problem at fixed energy if the electric and magnetic potentials 
 and their derivatives of all order are exponentially decreasing. In \cite{PSU} the authors prove that the electric and the magnetic potentials are uniquely determined
 by the scattering matrix at one fixed energy if these potentials and their first derivatives are exponentially decreasing. In \cite{WY} the authors reconstruct the 
 asymptotics of the electric potential and the magnetic field from the knowledge of the scattering matrix at one fixed energy in the regular case (i.e. if the potentials 
 are sums of homogeneous terms at infinity). Finally, in the very recent paper \cite{KU}, the authors study an inverse boundary problem for magnetic Sch\"odinger operators
 on compact Riemannian manifolds with boundary for bounded magnetic and electric potentials. In this work they study the case of admissible geometries (i.e. compact
 Riemannian manifolds with boundary which are conformally embedded in a product of the Euclidean line and a simple manifold) where they reconstruct the magnetic and 
 electric potentials from the knowledge of the Cauchy data on the boundary of the manifolds.
\item {\it Inverse scattering for magnetic Schr\"odinger operators in dimension two.} This is the framework we study in this paper.
In \cite{ER3} the authors solve the inverse scattering problem at high energies, in \cite{Wed} the author reconstructs the magnetic field and the magnetic 
flux modulo $2$ from the knowledge of the scattering matrix at all energies with or without obstacle. In \cite{Nic3} Nicoleau uses a stationary approach to 
determine the asymptotic of the scattering operator and he reconstructs the electric potential and the magnetic field from the first two terms of this asymptotic 
expansion. In \cite{Nic2} he solves the inverse scattering problem at all energies with Aharonov-Bohm effect for short-range potentials and compactly supported
magnetic fields.
\end{itemize}


\subsection{Overview of the proof}

The proof of Theorem \ref{mainthm} is divided into four steps that we describe here.\\

\noindent
\underline{Step 1:} The first step of the proof consists in solving the direct problem. This will be done in Section \ref{dirpb}. In this Section we first study the free
case, i.e. when the potential $q_{\nu}(r)$ is assumed to be zero for all $r \geq r_0$. In this case we obtain explicit expressions of the scattering objects.
Indeed, we can write the free Jost solutions in terms of Bessel functions for all $r \geq r_0$:
\[ F_0^+(r,\nu) = e^{i\left( \nu_R + \frac{1}{2} \right) \frac{\pi}{2}} \sqrt{\frac{\pi r}{2}} H_{\nu_R}^{(1)}(r), \quad \forall r \geq r_0\]
and
\[  F_0^-(r,\nu) = e^{-i\left( \nu_R + \frac{1}{2} \right) \frac{\pi}{2}} \sqrt{\frac{\pi r}{2}} H_{\nu_R}^{(2)}(r), \quad \forall r \geq r_0.\]
By definition, we can then also rewrite the free regular solution and the free Jost functions in terms of Bessel functions and we also obtain an explicit formula for the free Regge
interpolation function given by
\[ \sigma_0(\nu) = - e^{i \pi(\nu - \nu_R)} \frac{H_{\nu_R}^{(2)}(r_0)}{H_{\nu_R}^{(1)}(r_0)}.\]
The aim in the following of Section \ref{dirpb} is then to compare the general scattering objects with the free ones. First,
in Section \ref{secjostsol} we prove using good estimates on the Green kernel obtained in Section \ref{estigreen} that the Jost solutions are 
holomorphic functions with respect to $\nu$ on the whole complex plane $\C$ satisfying the following properties:
 \begin{enumerate}
  \item They are of order $1$ with infinite type.
  \item There exists a positive constant $C$ such that
 \[ F^{\pm}(r,\nu) \sim C F_0^{\pm}(r,\nu), \quad \text{as} \quad \nu \to + \infty.\]
  \item They are bounded on $i\R + \gamma(R)$.
 \end{enumerate}
 Secondly, in Section \ref{etudesolreg} we study the regular solution and we show that it is a holomorphic function with respect to $\nu \in \C$ such that there exists a
 positive constant $C$ such that for all $r_0 \leq r \leq R$ and for all $\text{Re}(\nu_R) = \text{Re}(\nu) - \gamma(R) \geq 0$
 \begin{equation}\label{inegsolreg}
 | \Phi(r,\nu) | \leq \frac{C}{1 + |\nu_R|} \left( \frac{r}{r_0} \right)^{\text{Re}(\nu_R)}.  
 \end{equation}
 Finally, in Section \ref{etudejostfun} we study the Jost functions $\alpha(\nu)$ and $\beta(\nu)$ and we show they
 are holomorphic functions of order $1$ with infinite type on the whole complex plane $\C$ such that if $(\alpha,\beta)$ and 
 $(\tilde{\alpha},\tilde{\beta})$ are two couples of Jost functions associated with two couples $(V,B)$ and $(\tilde{V},\tilde{B})$ belonging to the class
 $\mathcal{C}$ then, for all $\text{Re}(\nu_R) \geq 0$,
 \begin{equation}\label{idalgutile}
   \alpha(\nu)\tilde{\beta}(\nu) - \tilde{\alpha}(\nu) \beta(\nu) = \frac{1}{2i} \int_{r_0}^{+\infty} (q_{\nu}(r) - \tilde{q}_{\nu}(r)) \Phi(r,\nu) \tilde{\Phi}(r,\nu) dr.
 \end{equation}
 We finally show that
 \begin{equation}\label{resutileflux}
   \sigma(\nu) = e^{i \pi \left( \nu + \frac{1}{2} \right) } \frac{\alpha(\nu)}{\beta(\nu)} \to e^{-i \pi \gamma(R)}, \quad \nu \to + \infty,
 \end{equation}
 and this result is interesting because it was predicted by the work \cite{RY} of Roux and Yafaev.\\

\noindent
\underline{Step 2:} From now on, we solve the inverse problem and we thus introduce two couples $(V,B)$ and $(\tilde{V},\tilde{B})$ belonging to the class $\mathcal{C}$. The 
first step of the proof consists in reconstructing the magnetic flux. Precisely, we show, using (\ref{resutileflux}) and (\ref{defdelta1}), that if
\[ \delta_l = \tilde{\delta}_l, \quad \forall l \in \mathcal{L},\]
where $\mathcal{L} \subset \N$ satisfies the M\"untz condition, then
\[ \gamma(R) = \tilde{\gamma}(R), \quad \text{mod} \quad 2.\]
We then use the invariance of the Regge interpolation function with respect to translations of the magnetic flux by $2k$, $k \in \Z$, to conclude that
\[ \gamma(R) = \tilde{\gamma}(R).\]

\noindent
\underline{Step 3:} The third step consists in proving the uniqueness of the Regge interpolation function on almost the whole complex plane $\C$. To do this we use an idea dued
to Ramm \cite{Ram} and we thus introduce the function
 \[  F(\nu) = 2i(\alpha(\nu) \tilde{\beta}(\nu) - \tilde{\alpha}(\nu)\beta(\nu)).\]
Thanks to (\ref{idalgutile}) we know that
\[ F(\nu) = \int_{r_0}^{+ \infty} p_{\nu}(r) \Phi(r,\nu) \tilde{\Phi}(r,\nu) dr,\]
where
\[ p_{\nu}(r) = q_{\nu}(r) - \tilde{q}_{\nu}(r) = 0, \quad \forall r \geq R.\]
We then prove, using the estimate (\ref{inegsolreg}) we obtained before on the regular solution, that the function $F$ belongs to a particular class of functions, called the Nevanlinna
class $N(\Pi^+)$, when restricted to the half plane 
\[ \Pi^+ = \{ \nu_R \in \mathbb{C}, \, \mathrm{Re}(\nu_R) > 0 \}.\]
The result we are interested in on the Nevanlinna class is the following:

\begin{theorem}[\cite{Ram}, Thm. 1.3]\label{thmNev2}
 Let $h \in N(\Pi^+)$ satisfying $h(n) = 0$ for all $n \in \mathcal{L}$ where $\mathcal{L} \subset \mathbb{N}^{\star}$ with $\sum_{n \in \mathcal{L}} \frac{1}{n} = \infty$.
 Then $h \equiv 0$ in $\Pi^+$.
\end{theorem}
\noindent
 We recall that our main assumption is
 \begin{equation}\label{hypoprincdelta2}
   \delta_l = \tilde{\delta}_l, \quad \forall l \in \mathcal{L},
 \end{equation}
 where $\mathcal{L} \subset \N^{\star}$ satisfies
 \[ \sum_{l \in \mathcal{L}} \frac{1}{l} = + \infty.\]
 From (\ref{hypoprincdelta2}) we easily deduce that
 \[ \sigma(l) = e^{2i \delta_l} = e^{2i \tilde{\delta}_l} = \tilde{\sigma}(l), \quad \forall l \in \mathcal{L}.\]
 Therefore, since
 \[ \sigma(l) = e^{i \pi \left(l + \gamma(R) + \frac{1}{2} \right) } \frac{\alpha(l)}{\beta(l)} \quad \text{and} \quad 
  \tilde{\sigma}(l) = e^{i \pi \left(l + \tilde{\gamma}(R) + \frac{1}{2} \right) } \frac{\tilde{\alpha}(l)}{\tilde{\beta}(l)},\]
  using that we previously showed that $\gamma(R) = \tilde{\gamma}(R)$ and the definition of the function $F$, we obtain that
  \[ F(l) = 0, \quad \forall l \in \mathcal{L}.\]
  So, thanks to the previous Theorem \ref{thmNev2}, we obtain that the map $\nu \mapsto F(\nu)$ is identically zero on $\Pi^+$.
  Finally $F$ is identically zero on $\C$ since it is a holomorphic function. Therefore, for all $\nu \in \C$
 \[ F(\nu) = 2i(\alpha(\nu) \tilde{\beta}(\nu) - \tilde{\alpha}(\nu)\beta(\nu)) = 0\]
 and we can then conclude that
  \[ \sigma(\nu) = \sigma(\nu), \quad \forall \nu \in \C \setminus \{\beta(\nu) = 0 \}.\]

  \noindent
\underline{Step 4:} In the last step of the proof we use an argument close in spirit to the B\"org -Marchenko Theorem to conclude the proof of Theorem \ref{mainthm}.
We fix $r \geq r_0$ and we define $F(r,\nu)$
as an application of the complex variable $\nu$ by
\[ F(r,\nu) = F^+(r,\nu) \tilde{F}^-(r,\nu) - F^-(r,\nu) \tilde{F}^+(r,\nu),\]
where $F^{\pm}(r,\nu)$ and $\tilde{F}^{\pm}(r,\nu)$ are the Jost solutions associated with the potentials $q_{\nu}$ and $\tilde{q}_{\nu}$ respectively. We 
are able to prove thanks to our previous study of the Jost solutions
that this application is holomorphic on the whole complex plane $\C$ and
of order $1$ with infinite type. Moreover, $F(r,\nu)$ is bounded on the imaginary axis $i \R + \gamma(R)$ and using that
\[ F(r,\nu) = \tilde{\Psi}(r,\nu) F^+(r,\nu) - \Psi(r,\nu) \tilde{F}^+(r,\nu) + e^{-i \pi \left( \nu + \frac{1}{2} \right )}\left( \sigma(\nu) - \tilde{\sigma}(\nu)\right)F^+(r,\nu) \tilde{F}^+(r,\nu),\]
where
\[ \Psi(r,\nu) = \frac{\Phi(r,\nu)}{\beta(\nu)},\]
we can show, using the equality of the Regge interpolation functions on the real line and the previous estimates obtained in the direct scattering problem, that $F(r,\nu) \to 0$ when
$\nu \to + \infty$. Finally, using a symmetry property we also obtain that $F(r,\nu) \to 0$ as $\nu \to - \infty$.\\
So, using the Phragm\'en-Lindel\"of Theorem on each quadrant of the complex plane, we deduce that $F(r,\nu)$ is
bounded on the whole complex plane $\C$. Using the Louville's Theorem and the limit $F(r,\nu) \to 0$ as $\nu \to + \infty$ we conclude that $F(r,\nu)$ is
identically zero on the whole complex plane. In 
other words, we know that
\[ F^+(r,\nu) \tilde{F}^-(r,\nu) =  F^-(r,\nu) \tilde{F}^+(r,\nu), \quad \forall r \geq r_0, \quad \forall \nu \in \C.\]
From this equality we thus obtain that the Jost solutions are then so closed that
\[ q_{\nu}(r) = \tilde{q}_{\nu}(r), \quad \forall r \geq r_0, \quad \forall \nu \in \C.\]
Since this equality is satisfied for all $\nu$ we can thus decouple the potentials and we can then conclude that
\[ V(x) = \tilde{V}(x) \quad \text{and} \quad B(x) = \tilde{B}(x),\quad \forall |x| \geq r_0.\]
The same approach has been used recently to study inverse scattering problems on asymptotically hyperbolic manifolds (see \cite{Pap1,DKN,DN}). In the hyperbolic setting,
we can prove that the Jost solutions are perturbations of the modified Bessel
functions $I_{\nu}(z)$. Moreover, in the hyperbolic context, the variable $\nu$ is fixed and depends only on the geometry of the manifolds whereas the variable $z$ ranges over $\C$.
However, in the Euclidean setting of the paper, the situation is really different. Indeed, as in the hyperbolic case, the Jost solutions are close
to the Hankel functions $H_{\nu}^{(j)}(r)$, but the complex angular momentum $\nu$ can be as large as possible and the radial variable
$r$ ranges over the compact set $(r_0,R)$.

The paper is organized as follows. In Section \ref{dirpb} we will study the direct scattering problem. In Section \ref{invpb} we solve the inverse scattering problem
at fixed energy. Finally, in Appendix \ref{appe} we give some useful estimates on the Bessel functions and the corresponding Green kernels.

\section{Direct scattering}\label{dirpb}
\noindent
In this Section we study the direct scattering problem. It means that we will recall the definition and the basic properties of the scattering objects we are interested in.
Moreover, we will compare them with the ones we obtain in the free case.

\subsection{The free case}
\noindent
The aim of the Section is to study the free case and to give the explicit expressions of the scattering objects in this case.
We first recall that the free Hamiltonian is given by
\[ H_{\nu}^{0} = -\frac{d^2}{dr^2} + \frac{(\nu_R)^2 - \frac{1}{4}}{r^2}, \quad \forall r \geq r_0.\]
Therefore, the free version of Equation (\ref{eqstat}) is
\begin{equation}\label{eqstatlibre}
 -u'' +  \frac{(\nu_R)^2 -  \frac{1}{4}}{r^2}  u = u, \quad \forall r \geq r_0,
\end{equation}
which is a modified Bessel equation (see (5.4.11) in \cite{Leb})
\[ u'' + \frac{1-2\alpha}{X} u' + \left( (\beta \gamma X^{\gamma-1})^2 + \frac{\alpha^2 - \mu^2 \gamma^2}{X^2} \right) u = 0,\]
is we choose $\alpha = \frac{1}{2}$, $\gamma = 1$, $\beta = 1$ and $\mu = \nu_R$.
Thus the free Jost solutions, which are the solutions of (\ref{eqstatlibre}) satisfying
\[ F_0^{\pm}(r,\nu) \sim e^{\pm i r}, \quad r \to + \infty,\]
are given by
\begin{equation}\label{defF_0+}
F_0^+(r,\nu) = e^{i\left( \nu_R + \frac{1}{2} \right) \frac{\pi}{2}} \sqrt{\frac{\pi r}{2}} H_{\nu_R}^{(1)}(r), \quad \forall r \geq r_0 
\end{equation}
and
\begin{equation}\label{defF_0-}
 F_0^-(r,\nu) = e^{-i\left( \nu_R + \frac{1}{2} \right) \frac{\pi}{2}} \sqrt{\frac{\pi r}{2}} H_{\nu_R}^{(2)}(r), \quad \forall r \geq r_0.
\end{equation}

\begin{lemma}
 For every fixed $r \geq r_0$, the free Jost solutions $\nu \mapsto F_0^{\pm}(r,\nu)$ are holomorphic functions of order $1$ with infinite type with respect
 to $\nu \in \C$. Moreover, these functions are even with respect to $\nu_R = \nu - \gamma(R)$.
\end{lemma}

\begin{remark}\label{rkjostlibre}
 As it has been recalled in \cite{DN5}, for any fixed $r$, there exist positive constants $A$ and $B$ such that for all $\text{Re}(\nu_R) \geq 0$,
 \[ | F_0^{+}(r,\nu) | \leq A e^{B |\nu_R|} |\Gamma(\nu_R + 1)|.\]
 But, since we work for $r_0 \leq r \leq R$ we can actually find uniform positive constants $A$ and $B$ such that for all such $r$ and for all $\text{Re}(\nu_R) \geq 0$
 \[ | F_0^{+}(r,\nu) | \leq A e^{B |\nu_R|} |\Gamma(\nu_R + 1)|.\]
\end{remark}
\noindent
Moreover, we also know that the free regular solution, which is the solution of (\ref{eqstatlibre}) satisfying a Dirichlet condition at $r = r_0$, is given, up to a
multiplicative constant, by
\begin{eqnarray*}
 \Phi_0(r,\nu) &=& i\left( F_0^-(r_0,\nu) F_0^+(r,\nu) - F_0^+(r_0,\nu) F_0^-(r,\nu) \right)\\
 &=& i \frac{\pi \sqrt{r r_0}}{2} \left( H_{\nu_R}^{(2)}(r_0) H_{\nu_R}^{(1)}(r) - H_{\nu_R}^{(1)}(r_0) H_{\nu_R}^{(2)}(r) \right).
\end{eqnarray*}

\begin{prop}\label{estiphi0}
 There exists a positive constant $C$ such that for all $r_0 \leq r \leq R$ and for all $\text{Re}(\nu_R) \geq 0$
 \[ | \Phi_0(r,\nu) | \leq \frac{C}{|\nu_R|+1} \left( \frac{r}{r_0} \right)^{\text{Re}(\nu_R)}.\]
\end{prop}

\begin{proof}
 We note that for all $r \geq r_0$
 \[ \Phi_0(r,\nu) = \frac{1}{2} N(r,r_0,\nu),\]
 and we then conclude thanks to Proposition \ref{estiN1proof}.
\end{proof}
\noindent
Finally, we obtain that the free Jost functions, which are defined by
\[  \Phi_0(r,\nu) = \alpha_0(\nu) F_0^+(r,\nu) + \beta_0(\nu) F_0^-(r,\nu),\]
since the Jost solutions are a Fundamental System of Solutions of (\ref{eqstatlibre}), are given by
\[ \alpha_0(\nu) = i F_0^-(r_0,\nu) = ie^{-i\left( \nu_R + \frac{1}{2} \right) \frac{\pi}{2}} \sqrt{\frac{\pi r}{2}} H_{\nu_R}^{(2)}(r_0)\]
and
\begin{equation}\label{defbeta0}
 \beta_0(\nu) = -i F_0^+(r_0,\nu) = -i e^{i\left( \nu_R + \frac{1}{2} \right) \frac{\pi}{2}} \sqrt{\frac{\pi r}{2}} H_{\nu_R}^{(1)}(r_0) .
\end{equation}
Hence, the free Regge interpolation function is
\[ \sigma_0(\nu) = e^{i \pi \left(\nu + \frac{1}{2} \right) } \frac{\alpha_0(\nu)}{\beta_0(\nu)} = - e^{i \pi(\nu - \nu_R)} \frac{H_{\nu_R}^{(2)}(r_0)}{H_{\nu_R}^{(1)}(r_0)}.\]
In particular, for $l \in \Z$ we obtain the result given in \cite{Rui}, Eq. $(4.22)$, 
\[ \sigma_0(l) = - e^{i \pi \gamma(R)} \frac{H_{l-\gamma(R)}^{(2)}(r_0)}{H_{l-\gamma(R)}^{(1)}(r_0)}.\]
Therefore, using the asymptotics of the Bessel functions (\ref{asHnu1}) and (\ref{asHnu2}) and the symmetry identities (\ref{lienBessel1}) and (\ref{lienBessel2}), we obtain that
\[ \sigma_0(l) \to e^{i \pi \gamma(R)} \quad \text{as} \quad l \to + \infty\]
and
\[ \sigma_0(l) \to e^{-i \pi \gamma(R)} \quad \text{as} \quad l \to - \infty.\]

\subsection{Definition and properties of the Jost solutions}\label{secjostsol}
\noindent
In this Section we study the Jost solutions. First, let us recall the definition of these functions.

\begin{definition}[Jost solutions]
 The Jost solutions are the solutions of the stationary equation (\ref{eqstat}) satisfying the asymptotics
 \[ F^{\pm}(r,\nu) \sim e^{\pm ir}, \quad r \to + \infty.\]
\end{definition}

\begin{remark}
 These Jost solutions are in fact closely related to the ones introduced in \cite{DN5}. Precisely,
 \[ F^{\pm}(r,\nu) = f^{\pm}(r,\nu_R), \quad \forall r \geq r_0,\]
 where $f^{\pm}(r,\nu_R)$ are the Jost functions of \cite{DN5}.
\end{remark}
\noindent
We easily obtain the following Lemma which gives us useful properties on the Jost solutions.

\begin{lemma}\label{lemJost}
  For $r \geq r_0$ fixed, the Jost solutions are holomorphic functions with respect to $\nu$ satisfying the following identity:
  \[ \overline{F^{\pm}(r,\nu)} = F^{\mp}(r,\overline{\nu}).\]
  Moreover,
  \begin{equation}\label{lemmenonvanish}
 F^{\pm}(r,\nu) \neq 0, \quad \forall \nu \in \R, \quad \forall r \geq r_0.   
  \end{equation}
\end{lemma}

\begin{proof}
 Let us give the proof of the last point.
 Assume for instance that $F^+(r,\nu) = 0$ for some $r \geq r_0$. Since, $\overline{F^+(r,\nu)} = F^-(r,\nu)$ we also have $F^-(r,\nu) = 0$ which contradicts
\[ W(F^+(r,\nu),F^-(r,\nu)) = -2i.\]
\end{proof}
\noindent
Because of the presence of $\nu$ in the potential $q_{\nu}$ the Jost solutions are not even functions with respect to $\nu$ contrary to the ones introduced in the
non-magnetic case studied in \cite{DN5}. However, there is an important symmetry identity that will be usefull in the following.

\begin{lemma}\label{symjostsol}
 Let $F_{\gamma}^{\pm}(r,\nu)$ be the Jost solutions associated with the couple $(V,B) \in \mathcal{C}$. $F_{-\gamma}^{\pm}(r,\nu)$ are then the Jost solutions associated
 with the couple $(V,-B) \in \mathcal{C}$ and
 \begin{equation}\label{symetrieJost}
  F_{\gamma}^{\pm}(r,\nu) = F_{-\gamma}^{\pm}(r,-\nu), \quad \forall r \geq r_0, \quad \forall \nu \in \C.
  \end{equation}
\end{lemma}

\begin{proof}
 We first note that the Hamiltonian defined in (\ref{hampart}) satisfies
 \begin{equation}\label{egham}
 H_{\nu,\gamma} = H_{-\nu,-\gamma}.  
 \end{equation}
 We recall that the Jost solutions $F_{\gamma}^{\pm}(r,\nu)$ are the unique solutions of
\[H_{\nu,\gamma}F_{\gamma}^{\pm}(r,\nu) = F_{\gamma}^{\pm}(r,\nu)\]
satisying suitable asymptotics at infinity. From (\ref{egham}), we then deduce that
\[ H_{-\nu,-\gamma}F_{\gamma}^{\pm}(r,\nu) = F_{\gamma}^{\pm}(r,\nu)\]
and by uniqueness we can then conclude that
\[ F_{\gamma}^{\pm}(r,\nu) = F_{-\gamma}^{\pm}(r,-\nu).\]
\end{proof}

\begin{remark}
 We will use this equality to extend the results we can obtain on the right half-plane to the left half-plane. Indeed, if $(V,B) \in \mathcal{C}$ then $(V,-B) \in \mathcal{C}$.
 So, if we can prove a property for the Jost solutions $F_{\gamma}^{\pm}(r,\nu)$ on the right half-plane, this property is also true for $F_{-\gamma}^{\pm}(r,\nu)$ and thanks to previous Lemma we then
 obtain that this property is also true for $F_{\gamma}^{\pm}(r,-\nu)$, i.e. on the left half-plane.
\end{remark}
\noindent
The aim of this Section is to prove the following Proposition on the Jost solutions.

\begin{prop}\label{propjostsol}
 For $r \geq r_0$ fixed, the Jost solutions satisfy the following properties.
 \begin{enumerate}
  \item They are of order $1$ with infinite type on the whole complex plane $\C$.
    \item There exists a positive constant $C$ such that
 \[ F^{\pm}(r,\nu) \sim C F_0^{\pm}(r,\nu), \quad \text{as} \quad \nu \to + \infty.\]
  \item They are bounded on $i\R + \gamma(R)$.
 \end{enumerate}
\end{prop}
\noindent
We use the method of variation of constants which tells us
that for every $r \geq r_0$ the Jost solutions satisfy
\begin{equation}\label{eqintJost}
 F^{\pm}(r,\nu) = F_0^{\pm}(r,\nu) + \int_r^{+\infty} N(r,s,\nu) q_{\nu}(s)F^{\pm}(s,\nu) ds,
\end{equation}
where the Green kernel $N(r,s,\nu)$ is defined by
\[ N(r,s,\nu) = u(r) v(s) - u(s)v(r),\]
where $(u,v)$ is a Fundamental System of Solutions of (\ref{eqstat}) when $q_{\nu} = 0$ defined by
\[ u(r) = \sqrt{\frac{\pi r}{2}} J_{\nu_R}(r) \quad \text{and} \quad v(r) = -i \sqrt{\frac{\pi r}{2}} H_{\nu_R}^{(1)}(r).\]
First, we give an estimate of the kernel $N(r,s,\nu)$ which will be useful in the following.

\begin{lemma}\label{estiN1}
 There exists a constant $C > 0$ such that for all $\nu \in \{ \text{Re}(\nu_R) \geq 0 \} = \{ \text{Re}(\nu) \geq \gamma(R) \} $ and for all $r_0 \leq r \leq s \leq R$
 \[ |N(r,s,\nu)| \leq \frac{C}{|\nu_R|+1} \left( \frac{s}{r} \right)^{\text{Re}(\nu_R)}.\]
\end{lemma}

\begin{proof}
 See Proposition \ref{estiN1proof}.
\end{proof}

\begin{remark}
 In \cite{DN5} the authors introduced the Green kernel
 \[ M(r,s,\nu) = \frac{F_0^+(s,\nu)}{F_0^+(r,\nu)} N(r,s,\nu),\]
 but it forces to work only on particular regions of the complex plane because of the zeros of the Jost solution $F_0^+(r,\nu)$. Since in our case we work only for
 $r_0 \leq r \leq s \leq R$ we can introduce a new kernel which is slightly different and permits us to avoid this problem.
\end{remark}
\noindent
We introduce a new Green kernel
\[ M(r,s,\nu) = \left( \frac{r}{s} \right)^{\nu_R} N(r,s,\nu).\]
Thanks to Lemma \ref{estiN1} we immediately obtain the following Lemma.

\begin{lemma}\label{estiMJsol}
 There exists a constant  $C > 0$ such that for all $\nu \in \{ \text{Re}(\nu_R) \geq 0 \} = \{ \text{Re}(\nu) \geq \gamma(R) \}$ and for all $r_0 \leq r \leq s \leq R$:
 \[ |M(r,s,\nu)| \leq \frac{C}{|\nu_R|+1}.\]
\end{lemma}
\noindent
We now want to prove the following Lemma.

\begin{lemma}
 The Jost solutions $F^{\pm}(r,\nu)$ are holomorphic functions of order $1$ with infinite type for $\nu \in \{\text{Re}(\nu_R) \geq 0 \}$, i.e. $\nu \in \{\text{Re}(\nu) \geq \gamma(R) \}$.
\end{lemma}

\begin{proof}
We just give the proof for the Jost solution $F^+(r,\nu)$. We use an iterative method to prove this Lemma. First,
we multiply the integral equation (\ref{eqintJost}) by $\left( \frac{r}{R} \right)^{\nu_R}$ and we then obtain the new integral equation
\[ g(r,\nu) = \left( \frac{r}{R} \right)^{\nu_R} F_0^+(r,\nu) + \int_r^{+\infty} M(r,s,\nu) q_{\nu}(s)g(s,\nu) ds,\]
where
\[ g(r,\nu) = \left( \frac{r}{R} \right)^{\nu_R} F^+(r,\nu).\]
We then use an iterative method by putting
\[ g^0(r,\nu) = \left( \frac{r}{R} \right)^{\nu_R} F_0^+(r,\nu)\]
and
\[ g^{k+1}(r,\nu) = \int_r^{+\infty} M(r,s,\nu) q_{\nu}(s)g^k(s,\nu) ds, \quad \forall k \geq 0.\]
We then show using Lemma \ref{estiMJsol} and Remark \ref{rkjostlibre} that, for all $k \geq 0$, for all $\text{Re}(\nu_R) \geq 0$
\[ | g^k(r,\nu) | \leq A e^{B|\nu_R|} |\Gamma(\nu_R+1)| \frac{1}{k!} \left( \frac{C}{|\nu_R|+1} \int_r^{+\infty} |q_{\nu}(s)| ds \right)^k.\]
Therefore, for all $\text{Re}(\nu_R) \geq 0$,
\begin{eqnarray*}
 | g(r,\nu) | &\leq& \sum_{k \geq 0} |g^k(s,\nu)| \\
 &\leq&  A e^{B|\nu_R|} |\Gamma(\nu_R+1)| \exp \left( \frac{C}{|\nu_R|} \int_r^{+\infty} |q_{\nu}(s)| ds \right)
 \end{eqnarray*}
and finally, for all $\text{Re}(\nu_R) \geq 0$,
\[ | F^+(r,\nu) | \leq A e^{B|\nu_R|} |\Gamma(\nu_R+1)| \exp \left( \frac{C}{|\nu_R|} \int_r^{+\infty} |q_{\nu}(s)| ds \right)\left( \frac{R}{r_0} \right)^{\text{Re}(\nu_R)}.\]

\begin{remark}
 We note that the term
\[ A e^{B|\nu_R|} |\Gamma(\nu_R+1)|\]
corresponds to the contribution of $F_0^+(r,\nu)$ (see Remark \ref{rkjostlibre}).
\end{remark}
\noindent
From this inequality we deduce that for any fixed $r \in [r_0,R]$, the Jost solution $\nu \mapsto F^+(r,\nu)$ is a function 
of order $1$ with infinite type for all $\text{Re}(\nu_R) \geq 0$ thanks to the Stirling's formula (see \cite{Leb}, Eq (1.4.24)).
\end{proof}

\begin{coro}
  The Jost solutions $F^{\pm}(r,\nu)$ are holomorphic functions of order $1$ with infinite type for $\nu \in \C$.
\end{coro}

\begin{proof}
 Thanks to the previous Lemma we already know that the Jost solutions are holomorphic functions of order $1$ with infinite type for $\{\text{Re}(\nu_R) \geq 0 \}$,
 i.e. $\{\text{Re}(\nu) \geq \gamma(R) \}$. We now want to extend this property to the region $\{\text{Re}(\nu) \leq \gamma(R) \}$.
 In this aim, we use Lemma \ref{symjostsol} which tells us that
 \begin{equation}\label{symjostsol2}
  F_{\gamma}^{\pm}(r,\nu) = F_{-\gamma}(r,-\nu), \quad \forall r \geq r_0, \quad \forall \nu \in \C.
 \end{equation}
Indeed, since if $(V,B) \in \mathcal{C}$ then $(V,-B) \in \mathcal{C}$, we thus also know that the corresponding Jost solutions $F_{-\gamma}(r,\nu)$ are 
holomorphic functions of order $1$ with infinite type for $\{\text{Re}(\nu) \geq -\gamma(R) \}$. Moreover, if $\text{Re}(\nu) \leq \gamma(R)$, then
$\text{Re}(-\nu) \geq -\gamma(R)$. Therefore, thanks to (\ref{symjostsol2}), the Jost solutions $F_{\gamma}^{\pm}(r,\nu)$ are also 
 holomorphic functions of order $1$ with infinite type in the region $\{\text{Re}(\nu) \leq \gamma(R) \}$.
\end{proof}

\begin{prop}\label{lienFF_0}
 For any fixed $r \in [r_0,R]$, there exists a positive constant $C$ such that
 \[ F^+(r,\nu) \sim C F_0^+(r,\nu), \quad \text{as} \quad \nu \to + \infty.\]
\end{prop}

\begin{proof}
 Since $F_0^+(r,\nu) \neq 0$ for all $\nu \in \R$ (see Lemma \ref{lemJost}), we can here strictly follow the idea of \cite{DN5}, Lemma 4.4 and Proposition 4.5. We first rewrite the integral 
 equation (\ref{eqintJost}) as
 \[ h(r,\nu) = 1 + \int_r^{+\infty} K(r,s,\nu) q_{\nu}(s) h(s,\nu) ds,\]
 where
 \[ h(r,\nu) = \frac{F^+(r,\nu)}{F_0^+(r,\nu)}\]
 and
  \[ K(r,s,\nu) = \frac{F_0^+(s,\nu)}{F_0(r,\nu)} N(r,s,\nu).\]
  We note that thanks to Lemma \ref{estiK1proof}, for all $r_0 \leq r \leq s \leq R$, for $\nu_R \geq 0$,
  \[ K(r,s,\nu) = \frac{s}{2\nu_R+1} ( 1 + o(1)).\]
 We then use an iterative method to obtain, using Lemma \ref{estiK1proof}, that for all $\nu_R \geq 0$,
  \[  h(r,\nu) = \exp \left( \frac{-1}{2\nu_R + 1} \int_r^{+\infty} s q_{\nu}(s) ds \right) (1+o(1)).\]
We now recall that
\[ q_{\nu}(r) := - \frac{2\nu (\gamma(r)-\gamma(R))}{r^2} + \frac{\gamma(r)^2-\gamma(R)^2}{r^2} + V(r)\]
and we thus obtain that
  \begin{eqnarray*}
  h(r,\nu) &=& \exp \left( \frac{2\nu}{2\nu_R + 1} \int_r^{+\infty} \frac{\gamma(s)-\gamma(R)}{s} ds -  \frac{1}{2\nu_R + 1} \int_r^{+\infty} \frac{\gamma(s)^2-\gamma(R)^2}{s} + V(s) ds \right) (1+o(1))\\
  &=& \exp \left( \frac{2\nu}{2\nu_R + 1} \int_r^{+\infty} \frac{\gamma(s)-\gamma(R)}{s} ds \right) \exp \left( O \left( \frac{1}{\nu_R} \right) \right) (1+o(1)).
  \end{eqnarray*}
 Therefore,
 \[ h(r,\nu) \sim \exp \left( \int_r^{+\infty} \frac{\gamma(s)-\gamma(R)}{s} ds \right), \quad \text{as} \quad \nu \to + \infty.\]
 By definition, we can thus conclude that
 \[ F^+(r,\nu) \sim C_r F_0^+(r,\nu),\quad \text{as} \quad \nu \to + \infty,\]
 where
 \[ C_r := \exp \left( \int_r^{+\infty} \frac{\gamma(s)-\gamma(R)}{s} ds \right).\]
 \begin{remark}
  \[ C_r = 1, \quad \forall r \geq R.\]
 \end{remark}
\end{proof}

%

\begin{prop}\label{estiFiR}
 The Jost solutions $\nu \mapsto F^{\pm}(r,\nu)$ are bounded on $i\R + \gamma(R)$. In particular, there exists a positive constant $C$ such that for all
 $r_0 \leq r \leq R$ and for all $y \in \R$,
 \[ |F^{\pm}(r,iy+\gamma(R)) | \leq C \exp \left( \frac{C}{|y|+1} \int_r^{+\infty} |q_{\nu}(s)| ds \right).\]
\end{prop}

\begin{proof}
 We just give the proof for the Jost solution $F^{+}(r,\nu)$. We recall the integral equation on this Jost solution
 \[F^{+}(r,\nu) = F_0^{+}(r,\nu) + \int_r^{+\infty} N(r,s,\nu) q_{\nu}(s)F^{+}(s,\nu) ds.\]
 We then use an iterative method, for $\nu = iy+\gamma(R)$, i.e. $\nu_R = iy$, to obtain, using (\ref{estiH1iR}) and Proposition \ref{estiNiR}, that
 \[ |F^{+}(r,iy+\gamma(R)) | \leq C \exp \left( \frac{C}{|y|+1} \int_r^{+\infty} |q_{\nu}(s)| ds \right).\]
\end{proof}

\subsection{Study of the regular solution}\label{etudesolreg}
\noindent
In this Section we study the regular solution. First, let us recall the definition of this function.

\begin{definition}[Regular solution]
 The regular solution is the solution of the stationary equation (\ref{eqstat}) satisfying the Dirichlet condition
 \[ \Phi(r,\nu)_{|r = r_0} = 0.\]
\end{definition}
\noindent
By uniqueness and using the fact that the Jost solutions $F^{\pm}(r,\nu)$ are a Fundamental System of Solutions of (\ref{eqstat}) we obtain the following Lemma.

\begin{lemma}\label{devsolreg}
 The regular solution $\Phi(r,\nu)$ is given, up to a multiplicative constant, by
\[ \Phi(r,\nu) = i\left( F^-(r_0,\nu) F^+(r,\nu) - F^+(r_0,\nu) F^-(r,\nu) \right).\]
\end{lemma}
\noindent
We then easily obtain the following Lemma thanks to Lemma \ref{lemJost}.

\begin{lemma}
  For $r \geq r_0$ fixed, the regular solution is a holomorphic function with respect to $\nu$ on $\C$ satisfying the following identity:
  \[ \overline{\Phi(r,\nu)} = \Phi(r,\overline{\nu}).\]
\end{lemma}
\noindent
The aim of this section is then to prove the following Proposition on the regular solution.

\begin{prop}\label{estiphi}
There exists a positive constant $C$ such that for all $r_0 \leq r \leq R$ and for
all $\text{Re}(\nu_R) \geq 0$, i.e. for all $\text{Re}(\nu) \geq \gamma(R)$,
 \[ | \Phi(r,\nu) | \leq \frac{C}{1 + |\nu_R|} \left( \frac{r}{r_0} \right)^{\text{Re}(\nu_R)}.\]
\end{prop}

\begin{proof}
 We use the following integral equation for the regular solution for all $r \geq r_0$
 \[ \Phi(r,\nu) = \Phi_0(r,\nu) + \int_{r_0}^{r} N(r,s,\nu) q_{\nu}(s) \Phi(s,\nu) ds.\]
We then use an iterative method to obtain, using Propositions \ref{estiphi0} and \ref{estiN1proof} in the case where $s \leq r$, that for all
 $\text{Re}(\nu_R) \geq 0$,
 \[  | \Phi(r,\nu) | \leq \frac{C}{|\nu_R|+1} \left( \frac{r}{r_0} \right)^{\text{Re}(\nu_R)} \exp \left( \frac{C}{|\nu_R|+1} \int_{r_0}^r |q_{\nu}(s)| ds \right).\]
 We thus finally obtain that for all $r \geq r_0$ and for all $\text{Re}(\nu_R) \geq 0$
 \[ | \Phi(r,\nu) | \leq \frac{C}{1 + |\nu_R|} \left( \frac{r}{r_0} \right)^{\text{Re}(\nu_R)}.\]
\end{proof}

\subsection{Study of the Jost functions}\label{etudejostfun}
\noindent
The aim of this Section is to study the Jost functions and to obtain an algebraic identity on these functions that will be useful is the resolution of our inverse problem.
Let us first recall the definition of these functions.

\begin{definition}[Jost functions]
 The Jost functions are defined as the coefficients of the expansion of the regular solution on the Fundamental System of Solutions given by the Jost solutions, i.e.
\[ \Phi(r,\nu) = \alpha(\nu) F^+(r,\nu) + \beta(\nu) F^-(r,\nu).\]
\end{definition}
\noindent
By uniqueness and using Lemma \ref{devsolreg} we can actually show the following Lemma.

\begin{lemma}
The Jost solutions are given, up to a multiplicative constant, by
\[ \alpha(\nu) = iF^-(r_0,\nu)\]
and
\[ \beta(\nu) = -i F^+(r_0,\nu).\]
\end{lemma}
\noindent
Thanks to this Lemma we easily obtain the following Proposition.

\begin{prop}
 The Jost functions $\alpha(\nu)$ and $\beta(\nu)$ are holomorphic functions of order $1$ with infinite type on the whole complex plane $\C$.
\end{prop}

\begin{proof}
 We know that this is true for the Jost solutions (see Proposition \ref{propjostsol}) and we recall that
 \[ \alpha(\nu) = iF^-(r_0,\nu)\]
and
\[ \beta(\nu) = -i F^+(r_0,\nu).\]
\end{proof}

\begin{prop}\label{eqbeta}
 There exists a positive constant $C$ such that
 \[ |\beta(\nu)| \sim C|\beta_0(\nu)|, \quad \nu \to + \infty.\]
\end{prop}

\begin{proof}
 This is an easy consequence of Proposition \ref{lienFF_0} and the fact that
 \[ \beta(\nu) = -i F^+(r_0,\nu).\]
\end{proof}
\noindent
To solve our inverse problem we will need an algebraic identity on the Jost functions. Assume that $(\alpha,\beta)$ and $(\tilde{\alpha},\tilde{\beta})$ are the Jost functions corresponding to
two magnetic Schr\"odinger equations having the same magnetic flux, that is to say $\gamma(R) = \tilde{\gamma}(R)$.

\begin{lemma}
 If $\gamma(R) = \tilde{\gamma}(R)$ then, for all $\text{Re}(\nu_R) \geq 0$, i.e. for all $\text{Re}(\nu) \geq \gamma(R)$,
 \[ \alpha(\nu) - \tilde{\alpha}(\nu) = \frac{1}{2i} \int_{r_0}^{+\infty} (q_{\nu}(r) - \tilde{q}_{\nu}(r)) F^-(r,\nu) \tilde{\Phi}(r,\nu) dr\]
 and
 \[ \beta(\nu) - \tilde{\beta}(\nu) = - \frac{1}{2i} \int_{r_0}^{+\infty} (q_{\nu}(r) - \tilde{q}_{\nu}(r)) F^+(r,\nu) \tilde{\Phi}(r,\nu) dr.\]
\end{lemma}

\begin{proof}
We just give the proof of the first identity. First, using Equation (\ref{eqstat}), we know that for all $r \geq r_0$
 \[ (F^-(r,\nu) \tilde{\Phi}'(r,\nu) - F^{-'}(r,\nu) \tilde{\Phi}(r,\nu))' = (\tilde{q}_{\nu}(r) - q_{\nu}(r)) F^-(r,\nu) \tilde{\Phi}(r,\nu),\]
 since $\gamma(R) = \tilde{\gamma}(R)$ and so $\nu_R = \nu - \gamma(R) = \tilde{\nu}_R$.
 We integrate this identity on $[r_0,+\infty)$ and we obtain
 \[ \left[ W(F^-(r,\nu),\tilde{\Phi}(r,\nu)) \right]_{r_0}^{+\infty} = \int_{r_0}^{+\infty} (\tilde{q}_{\nu}(r) - q_{\nu}(r)) F^-(r,\nu) \tilde{\Phi}(r,\nu) dr.\]
 When $r \to +\infty$,
 \[ F^-(r,\nu) \sim \tilde{F}^-(r,\nu),\]
 thus,
 \[ W(F^-(r,\nu),\tilde{\Phi}(r,\nu)) \sim W(\tilde{F}^-(r,\nu),\tilde{\Phi}(r,\nu)) = 2i \tilde{\alpha}(\nu).\]
 When $r = r_0$,
 \[ F^-(r,\nu) \tilde{\Phi}'(r,\nu) - F^{-'}(r,\nu) \tilde{\Phi}(r,\nu) = F^-(r_0,\nu) \tilde{\Phi}'(r_0,\nu)\]
 since
 \[ \Phi(r_0,\nu) = 0.\]
 Moreover,
 \begin{eqnarray*}
  \tilde{\Phi}'(r_0,\nu) &=&  i\left( F^-(r_0,\nu) F^{+'}(r,\nu) - F^+(r_0,\nu) F^-(r,\nu) \right)_{| r = r_0}\\
  &=& i W(F^-(r,\nu), F^{+}(r,\nu))_{| r=r_0} = -2.
 \end{eqnarray*}
 We thus obtain, since $F^-(r_0,\nu) = -i\alpha(\nu)$,
 \[ 2i\alpha(\nu) - 2i\tilde{\alpha}(\nu) = \int_{r_0}^{+\infty} (q_{\nu}(r) - \tilde{q}_{\nu}(r)) F^-(r,\nu) \tilde{\Phi}(r,\nu) dr.\]
\end{proof}

\begin{prop}[\cite{DN5}, Proposition 5.6]\label{idalg}
 If $\gamma(R) = \tilde{\gamma}(R)$ then, for all $\text{Re}(\nu_R) \geq 0$, i.e. for all $\text{Re}(\nu) \geq \gamma(R)$,
 \[ \alpha(\nu)\tilde{\beta}(\nu) - \tilde{\alpha}(\nu) \beta(\nu) = \frac{1}{2i} \int_{r_0}^{+\infty} (q_{\nu}(r) - \tilde{q}_{\nu}(r)) \Phi(r,\nu) \tilde{\Phi}(r,\nu) dr.\]
\end{prop}
%
\begin{definition}[Regge interpolation function]\label{defReggedef}
 The Regge interpolation function is defined by
\begin{equation}\label{defRegge}
\sigma(\nu) = e^{i \pi \left( \nu + \frac{1}{2} \right)} \frac{\alpha(\nu)}{\beta(\nu)}.
\end{equation}
\end{definition}

\begin{prop} \label{limitsigma}
 \[ \sigma(\nu) \to e^{-i \pi \gamma(R)}, \quad \nu \to + \infty.\]
\end{prop}

\begin{proof}
 Thanks to the definition of the Jost functions
 \[ \alpha(\nu) = i F^-(r_0,\nu) \quad \text{and} \quad \beta(\nu) = -i F^+(r_0,\nu)\]
 we know that
 \[ \sigma(\nu) = - e^{i \pi \left( \nu + \frac{1}{2} \right)} \frac{F^-(r_0,\nu)}{F^+(r_0,\nu)}.\]
 We recall that thanks to Proposition \ref{propjostsol}
 \[ F^{\pm}(r_0,\nu) \sim C F_0^{\pm}(r_0,\nu), \quad \nu \to + \infty.\]
 Therefore, thanks to (\ref{defF_0+}-\ref{defF_0-}), when $\nu \to + \infty$,
 \begin{eqnarray*}
  \sigma(\nu) &=& - e^{i \pi \left( \nu + \frac{1}{2} \right)} \frac{F^-(r_0,\nu)}{F^+(r_0,\nu)}\\
  &\sim& - e^{i \pi \left( \nu + \frac{1}{2} \right)} \frac{CF_0^-(r_0,\nu)}{CF_0^+(r_0,\nu)}\\
  &=& - e^{-i \pi \gamma(R)} \frac{ H_{\nu_R}^{(2)}(r)}{ H_{\nu_R}^{(1)}(r)} = e^{-i \pi \gamma(R)},
 \end{eqnarray*}
 thanks to the identity (\ref{lienBessel1}).
\end{proof}

\begin{remark}
 We obtain that the eigenvalues of the scattering matrix $\sigma(l) = e^{2i \delta_l}$ concentrate to $e^{-i \pi \gamma(R)}$ when $l \to +\infty$. Actually, we can also 
 show that $\sigma(l)$ concentrate to $e^{i \pi \gamma(R)}$ when $l \to -\infty$. We thus obtain in our framework the result of \cite{RY}, Theorem 2.1 and Corollary 2.2 where the
 authors show that the essential spectrum of the scattering matrix consists in two complex conjugate points of the unit circle.
\end{remark}

\section{Inverse scattering}\label{invpb}
\noindent
In this Section we want to prove Theorem \ref{mainthm} and we thus introduce $(V,B) \in \mathcal{C}$ and $(\tilde{V},\tilde{B}) \in \mathcal{C}$
as radial and compactly supported electric potentials and magnetic fields respectively. Let $\mathcal{L} \subset \N^{\star}$ a subset satisfying the M\"untz condition
 \[ \sum_{n \in \mathcal{L}} \frac{1}{n} = + \infty.\]
 We assume that
 \[ \delta_l = \tilde{\delta}_l, \quad \forall l \in \mathcal{L}.\]
 We then want to show that
 \[ V(x) = \tilde{V}(x) \quad \text{and} \quad B(x) = \tilde{B}(x), \quad \forall |x| \geq r_0.\]

\subsection{Reconstruction of the magnetic flux}\label{secrecupflux}

\begin{theorem}\label{recupflux}
 Assume that
 \[ \delta_l = \tilde{\delta}_l, \quad \forall l \in \mathcal{L},\]
 where $\mathcal{L} \subset \N^{\star}$ satisfies the M\"untz condition $\sum_{l \in \mathcal{L}} \frac{1}{l} = + \infty$. Then,
 \[ \gamma(R) = \tilde{\gamma(R)}.\]
\end{theorem}

\begin{proof}
First, we recall that if
\[ \delta_l = \tilde{\delta}_l\]
then
\[ \sigma(l) = e^{2i \delta_l} = e^{2i \tilde{\delta}_l} = \tilde{\sigma}(l).\]
 Moreover, we know thanks to Proposition \ref{limitsigma} that
 \[ \sigma(l) \to e^{-i \pi \gamma(R)}, \quad l \to + \infty.\]
 Therefore, we obtain that
 \[ e^{-i \pi \gamma(R)} = e^{- i \pi \tilde{\gamma}(R)},\]
 so
 \[ \gamma(R) = \tilde{\gamma}(R)  \quad \text{mod} \quad 2.\]
At this point we only reconstruct the magnetic flux up to an even integer $2k$. Let us explain why we can actually assume that $2k = 0$ with no loss of generality. First,
note that 
\[ H_{\nu,\gamma+2k} = H_{\nu-2k,\gamma}.\]
In particular,
\[ \alpha_{\gamma + 2k}(\nu) = \alpha_{\gamma}(\nu - 2k)\]
and
\[ \beta_{\gamma + 2k}(\nu) = \beta_{\gamma}(\nu - 2k).\]
Hence, by Definition \ref{defReggedef}, we obtain that 
\[ \sigma_{\gamma + 2k}(\nu) = \sigma_{\gamma}(\nu - 2k).\]
Therefore, even if it means to change our main assumption in
\[ \delta_{l + 2k} = \tilde{\delta}_l, \quad \forall l \in \mathcal{L},\]
where $\mathcal{L} \subset \N^{\star}$ satisfies the M\"untz condition and $2k$, $k \in \Z$, is the difference of the magnetic flux we can assume with no loss of generality
that $2k = 0$.
\end{proof}


\subsection{First part of the proof: extension of the assumption}\label{proofp1}
\noindent
In this Section we follow an idea introduced by Ramm in \cite{Ram} to obtain the uniqueness of the Regge interpolation function on almost the whole complex plane.
In this aim, we consider the function
 \begin{equation}\label{defF}
  F(\nu) = 2i(\alpha(\nu) \tilde{\beta}(\nu) - \tilde{\alpha}(\nu)\beta(\nu)).
 \end{equation}
 As we proved previously in Proposition \ref{idalg} (using Theorem \ref{recupflux}), we know that
\[ F(\nu) = \int_{r_0}^{+ \infty} p_{\nu}(r) \Phi(r,\nu) \tilde{\Phi}(r,\nu) dr,\]
where
\[ p_{\nu}(r) = q_{\nu}(r) - \tilde{q}_{\nu}(r) = 0, \quad \forall r \geq R.\]
We now prove that the function $F$ belongs to a particular class of holomorphic functions when restricted to the half plane 
\[\Pi^+ = \{ z \in \mathbb{C}, \, \mathrm{Re}(z) > \gamma(R) \}.\]
Recall first the definition of the Nevanlinna class $N(\Pi^+)$.

\begin{definition}[Nevanlinna class, see \cite{Rudin} p.311]
 The Nevanlinna class $N(\Pi^+)$ is defined as the set of all analytic functions $f$ on $\Pi^+$ that satisfy the estimate
\[ \underset{0 < r < 1}{\sup} \int_{-\pi}^{\pi} \ln^+ \left| f \left( \frac{1-re^{i\varphi}}{1+re^{i\varphi}} \right) \right| \, \mathrm d\varphi < \infty,\]
where $\ln^+(x) = \ln(x)$ if $\ln(x) \geq 0$ and $0$ if $\ln(x) < 0$.
\end{definition}
\noindent
We shall use the following result proved in \cite{Ram}.

\begin{lemma}\label{Nev}
 Let $h \in H(\Pi^+)$ be a holomorphic function in $\Pi^+$ satisfying
\[\vert h(z) \vert \leq C e^{A \mathrm{Re}(z)}, \quad \forall z \in \Pi^{+},\]
where $A$ and $C$ are two constants. Then $h \in N(\Pi^+)$.
\end{lemma}
\noindent
As a consequence of Lemma \ref{Nev}, we get the following result.

\begin{coro}\label{Gnev}
 The map $\nu \mapsto F(\nu)$ belong to $N(\Pi^+)$.
\end{coro}


\begin{proof}
 We know thanks to Proposition \ref{estiphi} that
\[ |\Phi(r,\nu)| \leq \frac{C}{|\nu_R|+1} \left( \frac{r}{r_0} \right)^{\text{Re}(\nu_R)}, \quad \forall \text{Re}(\nu) \geq \gamma(R).\]
 Using this inequality, we obtain that
 \begin{eqnarray*}
  |F(\nu)| &\leq& \int_{r_0}^{R} |p_{\nu}(r)| |\Phi(r,\nu)| |\tilde{\Phi}(r,\nu)| dr \\
  &\leq& \left( \frac{C}{|\nu_R|+1} \right)^2 \int_{r_0}^{R} |p_{\nu}(r)| \left( \frac{r}{r_0} \right)^{2\text{Re}(\nu_R)}  dr.
 \end{eqnarray*}
 Therefore, since $r \in [r_0,R]$,
we obtain
 \begin{equation}\label{estiF}
  |F(\nu)| \leq \frac{C^2|\nu|}{(|\nu_R|+1)^2} \left( \frac{R}{r_0} \right)^{2\text{Re}(\nu_R)}.
 \end{equation}
 From this last inequality we deduce that there exist two positive constants $A$ and $C$ such that
 \[ | F(\nu)| \leq Ce^{A \text{Re}(\nu_R)}, \quad \forall \text{Re}(\nu) \geq \gamma(R).\]
\end{proof}
\noindent
We now recall the following result proved in \cite{Ram}, Theorem 1.3.

\begin{theorem}[\cite{Ram}, Thm. 1.3]\label{thmNev}
 Let $h \in N(\Pi^+)$ satisfying $h(n) = 0$ for all $n \in \mathcal{L}$ where $\mathcal{L} \subset \mathbb{N}^{\star}$ with $\sum_{n \in \mathcal{L}} \frac{1}{n} = \infty$.
 Then $h \equiv 0$ in $\Pi^+$.
\end{theorem}

\begin{coro}\label{uniF}
 The map $\nu \mapsto F(\nu)$ is identically zero on $\C$.
\end{coro}

\begin{proof}
 We recall that our main assumption is
 \begin{equation}\label{hypoprincdelta}
   \delta_l = \tilde{\delta}_l, \quad \forall l \in \mathcal{L},
 \end{equation}
 where $\mathcal{L} \subset \N^{\star}$ satisfies
 \[ \sum_{l \in \mathcal{L}} \frac{1}{l} = + \infty.\]
 From (\ref{hypoprincdelta}) we easily deduce that
 \[ \sigma(l) = e^{2i \delta_l} = e^{2i \tilde{\delta}_l} = \tilde{\sigma}(l).\]
 Therefore, since
 \[ \sigma(l) = e^{i \pi \left(l + \gamma(R) + \frac{1}{2} \right) } \frac{\alpha(l)}{\beta(l)} \quad \text{and} \quad 
  \tilde{\sigma}(l) = e^{i \pi \left(l + \tilde{\gamma}(R) + \frac{1}{2} \right) } \frac{\tilde{\alpha}(l)}{\tilde{\beta}(l)},\]
  using Theorem \ref{recupflux} and the definition of the function $F$ given by (\ref{defF}) we obtain that
  \[ F(l) = 0, \quad \forall l \in \mathcal{L}.\]
  Using Corollary \ref{Gnev} and Theorem \ref{thmNev} we thus obtain that the map $\nu \mapsto F(\nu)$ is identically zero on $\Pi^+$.
  Finally $F$ is identically zero on the whole complex plane $\C$ since it is a holomorphic function.
\end{proof}

\begin{theorem}\label{uniRegge}
 Assume that
  \[ \delta_l = \tilde{\delta}_l, \quad \forall l \in \mathcal{L}, \]
 where $\mathcal{L} \subset \N^{\star}$ satisfies
 \[ \sum_{l \in \mathcal{L}} \frac{1}{l} = + \infty.\]
 Then,
 \[ \sigma(\nu) = \tilde{\sigma}(\nu), \quad \forall \nu \in \C \setminus \{\beta(\nu) = 0 \}.\]
\end{theorem}

\begin{proof}
 Thanks to Corollary \ref{uniF} we know that for all $\nu \in \C$
 \[ F(\nu) = 2i(\alpha(\nu) \tilde{\beta}(\nu) - \tilde{\alpha}(\nu)\beta(\nu)) = 0.\]
 Thus, using the definition of the Regge interpolation function $\sigma$ given by (\ref{defRegge}), we obtain the result.
\end{proof}

\subsection{Second part of the proof: conclusion}
\noindent
We will follow the elegant,
simple and self-contained proof given in \cite{DN5}. This proof follows an idea close to the local B\"org-Marchenko uniqueness Theorem (see \cite{Be,GS,Si,Te}).
Let $r \geq r_0$ fixed, we introduce for $\nu \in \C$ the function 
\begin{equation}\label{defFfinal}
 F(r,\nu) = F^+(r,\nu) \tilde{F}^-(r,\nu) - F^-(r,\nu)\tilde{F}^+(r,\nu). 
\end{equation}
As we have seen in Section \ref{secjostsol}, Lemma \ref{lemJost}, the function $\nu \mapsto F(r,\nu)$ is holomorphic on $\C$ and of order $1$ with infinite type.
The aim is now to prove that this function is identically zero using the Phragm\'en-Lindel\"of Theorem (see for instance \cite{Bo}, Theorem 1.4.2).\\
\\
\noindent
\underline{Step 1: $F(r,\nu)$ is bounded on $i\R+\gamma(R)$.} We proved in Section \ref{secjostsol}, Proposition \ref{estiFiR}, that the Jost solutions $\nu \mapsto F^{\pm}(r,\nu)$ are bounded on
$i\R+\gamma(R)$. Therefore, $F(r,\nu)$ is, by definition, also bounded on $i\R+\gamma(R)$.\\
\\
\noindent
\underline{Step 2: $F(r,\nu) \to 0$ when $\nu \to + \infty$.} We strictly follow the proof given in \cite{DN5}. Indeed, for $\nu \in \R$, we know that
$\beta(\nu) \neq 0$ (see Lemma \ref{lemmenonvanish}) and we can then set
\[ \Psi(r,\nu) = \frac{\Phi(r,\nu)}{\beta(\nu)}.\]
We thus obtain, using (\ref{defsolreg}), that
\[ F^-(r,\nu) = \Psi(r,\nu) - \frac{\alpha(\nu)}{\beta(\nu)}F^+(r,\nu).\]
Therefore,
\[ F(r,\nu) = \tilde{\Psi}(r,\nu) F^+(r,\nu) - \Psi(r,\nu) \tilde{F}^+(r,\nu) + \left( \frac{\alpha(\nu)}{\beta(\nu)} - \frac{\tilde{\alpha}(\nu)}{\tilde{\beta}(\nu)} \right)F^+(r,\nu) \tilde{F}^+(r,\nu).\]
So, using (\ref{defRegge}), we deduce
\[ F(r,\nu) = \tilde{\Psi}(r,\nu) F^+(r,\nu) - \Psi(r,\nu) \tilde{F}^+(r,\nu) + e^{-i \pi \left( \nu + \frac{1}{2} \right )}\left( \sigma(\nu) - \tilde{\sigma}(\nu)\right)F^+(r,\nu) \tilde{F}^+(r,\nu).\]
Hence, by Theorem \ref{uniRegge}, we obtain that for $\nu \in \R$,
\[ F(r,\nu) = \tilde{\Psi}(r,\nu) F^+(r,\nu) - \Psi(r,\nu) \tilde{F}^+(r,\nu).\]
For instance, let us examine $\Psi(r,\nu) \tilde{F}^+(r,\nu)$. Propositions \ref{lienFF_0} and \ref{eqbeta} imply that for $\nu \geq \gamma(R)$
\[ |\Psi(r,\nu) \tilde{F}^+(r,\nu)| = \left| \frac{\Phi(r,\nu)}{\beta(\nu)} \right| |\tilde{F}^+(r,\nu)| \leq C \left| \frac{\Phi(r,\nu)}{\beta_0(\nu)} \right| |\tilde{F}_0^+(r,\nu)|.\]
Moreover, using Proposition \ref{estiphi}, we know that there exists a positive constant $C$ such that for a fixed $r \geq r_0$ and for all $\nu \geq \gamma(R)$
\[ | \Phi(r,\nu)| \leq \frac{C}{|\nu_R|+1} \left( \frac{r}{r_0} \right)^{\nu_R},\]
thanks to (\ref{defF_0+}) and (\ref{asHnu1})
\[ |F_0^+(r,\nu) | \leq C \left( \frac{r}{2} \right)^{-\nu_R + \frac{1}{2}} \Gamma(\nu_R),\]
and thanks to (\ref{defbeta0})
\[ |\beta_0(\nu)| \sim C \left( \frac{r_0}{2} \right)^{-\nu_R + \frac{1}{2}} \Gamma(\nu_R).\]
Therefore,
\[ |\Psi(r,\nu) \tilde{F}^+(r,\nu)| \leq \frac{C}{|\nu_R|+1} \to 0, \quad \nu \to +\infty.\]
Similarly, we obtain
\[ |\tilde{\Psi}(r,\nu) F^+(r,\nu)|  \to 0, \quad \nu \to +\infty.\]
Therefore,
\[F(r,\nu) \to 0, \quad \nu \to + \infty.\]
\\
\noindent
\underline{Step 3: $F(r,\nu) \to 0$ when $\nu \to - \infty$.} Thanks to the uniquenesss of the Regge interpolation function on the real line we know that for all
$\nu \in \R$
\[ F(r,\nu) = \tilde{\Psi}(r,\nu) F^+(r,\nu) - \Psi(r,\nu) \tilde{F}^+(r,\nu).\]
We then use Lemma \ref{symjostsol} to come back to the limit as $\nu \to + \infty$. Indeed, thanks to this result we obtain that, for large positive $\nu$
\[ F_{\gamma}^+(r,-\nu) = F_{-\gamma}^+(r,\nu), \]
\[ \beta_{\gamma}(-\nu) = \beta_{-\gamma}(\nu), \]
and
\[ \Phi_{\gamma}(r,-\nu) = \Phi_{-\gamma}(r,\nu).\]
Moreover, since if $(V,B) \in \mathcal{C}$ then $(V,-B) \in \mathcal{C}$, Propositions \ref{lienFF_0}, \ref{eqbeta} and \ref{estiphi} are also true for $(V,-B)$, i.e. 
for $-\gamma$. Therefore, we can follow the end of the previous step to obtain that 
\[F(r,\nu) \to 0, \quad \nu \to - \infty.\]
\\
\noindent
\underline{Step 4: $F(r,\nu)$ is identically zero on $\C$.} At this point, we know that $F(r,\nu) \to 0$ when $\nu \to \pm \infty$. In particular, $\nu \mapsto F(r,\nu)$ 
is bounded on the real axis. Since $F(r,\nu)$ is also bounded in $i\R+\gamma(R)$, applying the Phragm\'en-Lindel\"of Theorem in each quadrant of the complex
plane, we see that $F(r,\nu)$ is bounded on the whole complex place $\C$.
Finally, $F(r,\nu)$ is constant by Liouville's Theorem. As the limit is $0$ when $\nu \to +\infty$ we have
\begin{equation}\label{fidnulle}
F(r,\nu) = 0, \quad \forall \nu \in \C. 
\end{equation}
\\
\noindent
\underline{Step 5: Conclusion of the proof.} Using (\ref{defFfinal}) and (\ref{fidnulle}) we have
\begin{equation}\label{egpr1}
 F^+(r,\nu) \tilde{F}^-(r,\nu) = F^-(r,\nu)\tilde{F}^+(r,\nu), \quad \forall r \geq r_0, \quad \forall \nu \in \C.
\end{equation}
For $\nu \in \R$ fixed, we note that for all $r \geq r_0$, $F^{\pm}(r,\nu) \neq 0$ (see Lemma \ref{lemmenonvanish}).
We can then rewrite (\ref{egpr1}) as
\[ \frac{F^+(r,\nu)}{F^-(r,\nu)} = \frac{\tilde{F}^+(r,\nu)}{\tilde{F}^-(r,\nu)}, \quad \forall r \geq r_0, \quad \forall \nu \in \R.\]
Differentiating and using that $W(F^+(r,\nu),F^-(r,\nu)) = -2i$, it follows that
\[ F^-(r,\nu)^2 = \tilde{F}^-(r,\nu)^2.\]
We take the logarithmic derivative of this and we differentiate once more. We then obtain
\[ \frac{F^{-''}(r,\nu)}{F^-(r,\nu)} = \frac{\tilde{F}^{-''}(r,\nu)}{\tilde{F}^-(r,\nu)}, \quad \forall r \geq r_0, \quad \forall \nu \in \R.\]
Using (\ref{eqstat}) we then obtain that
\begin{equation}\label{egapot}
 q_{\nu}(r) = \tilde{q}_{\nu}(r),  \quad \forall r \geq r_0, \quad \forall \nu \in \R. 
\end{equation}
We now recall that for all $r \geq r_0$,
\[ q_{\nu}(r) = - \frac{2\nu (\gamma(r)-\gamma(R))}{r^2} + \frac{\gamma(r)^2-\gamma(R)^2}{r^2} + V(r)\]
and
\[ \tilde{q}_{\nu}(r) = - \frac{2\nu (\tilde{\gamma}(r)-\tilde{\gamma}(R))}{r^2} + \frac{\tilde{\gamma}(r)^2-\tilde{\gamma}(R)^2}{r^2} + \tilde{V}(r).\]
Since the equality (\ref{egapot}) is satisfied for all $\nu \in \R$, we can decouple it and we thus obtain the identities
\[ \gamma(r)-\gamma(R) = \tilde{\gamma}(r)-\tilde{\gamma}(R),  \quad \forall r \geq r_0\]
and
\[\frac{\gamma(r)^2-\gamma(R)^2}{r^2} + V(r) = \frac{\tilde{\gamma}(r)^2-\tilde{\gamma}(R)^2}{r^2} + \tilde{V}(r),  \quad \forall r \geq r_0.\]
Moreover, we already shown in Theorem \ref{recupflux} that
\[ \gamma(R) = \tilde{\gamma}(R).\]
Therefore, we can conclude that
\[ \gamma(r) = \tilde{\gamma}(r),  \quad \forall r \geq r_0,\]
which implies by differentiation (see (\ref{lienAgamma})) that
\[ B(x) = \tilde{B}(x), \quad \forall |x| \geq r_0.\]

\begin{remark}
The equality
\[ \gamma(r) = \tilde{\gamma}(r),  \quad \forall r \geq r_0,\]
actually allows us to conclude that
 \[ A(x) = \tilde{A}(x), \quad \forall |x| \geq r_0,\]
by (\ref{defAgamma}). This is not surprising that we can actually recover the magnetic potential $A$ because of the gauge choice we made in Definition \ref{defclassA}.
\end{remark}
%
\noindent
Finally, we also obtain that
\[ V(r) = \tilde{V}(r),  \quad \forall r \geq r_0.\]

\appendix

\section{Estimates on Bessel functions}\label{appe}

\subsection{Definition and symmetry properties of the Bessel functions}
\noindent
In this Section we recall the definition of the Bessel functions and we give some basic properties of these functions. We refer the reader to \cite{Leb}, Chapter 5, or to 
the Watson's treatise \cite{Wat} that we will use in the following.\\
\noindent
The Bessel function of the first kind $J_{\nu}(z)$ is defined for $\nu \in \C$ and $|\text{Arg}(z)| < \pi$ by
\[ J_{\nu}(z) = \sum_{k=0}^{+\infty} \frac{(-1)^k \left( \frac{z}{2} \right)^{\nu + 2k}}{\Gamma(k+1) \Gamma(k+\nu+1)}.\]
The Bessel function of the second kind are then defined for $\nu \in \C \setminus \Z$ and $|\text{Arg}(z)| < \pi$ by
\[ Y_{\nu}(z) = \frac{J_{\nu}(z) \cos(\nu \pi) - J_{-\nu}(z)}{\sin(\nu \pi)}.\]
Finally, the Bessel functions of the third kind, or Hankel functions, denoted by $H_{\nu}^{(1)}(z)$ and $H_{\nu}^{(2)}(z)$ are defined in terms of the Bessel functions of the 
first and the second kind by
\[ H_{\nu}^{(1)}(z) = J_{\nu}(z) + iY_{\nu}(z)\]
and
\[ H_{\nu}^{(2)}(z) = J_{\nu}(z) - iY_{\nu}(z).\]
These functions can be written as
\[ H_{\nu}^{(1)}(z) = \frac{J_{-\nu}(z)-e^{-i \pi \nu} J_{\nu}(z)}{i \sin(\nu \pi)}\]
and
\[ H_{\nu}^{(2)}(z) = \frac{e^{i \pi \nu} J_{\nu}(z)-J_{-\nu}(z)}{i \sin(\nu \pi)}.\]
The Bessel functions $J_{\nu}(z)$ and the Hankel functions $H_{\nu}^{(i)}(z)$ are entire functions with respect to $\nu$. Moreover, they satisfy the following identities 
(see \cite{MOS}, p.66)
\begin{equation}\label{lienBessel1}
 \overline{J_{\nu}(z)} = J_{\overline{\nu}}(\overline{z}), \quad \overline{Y_{\nu}(z)} = Y_{\overline{\nu}}(\overline{z}), \quad \overline{H_{\nu}^{(1)}(z)} = H_{\overline{\nu}}^{(2)}(\overline{z}),
\end{equation}
\begin{equation}\label{lienBessel2}
 H_{-\nu}^{(1)}(z) = e^{i \pi \nu} H_{\nu}^{(1)}(z) \quad \text{and} \quad H_{-\nu}^{(2)}(z) = e^{-i \pi \nu} H_{\nu}^{(2)}(z).
\end{equation}

\subsection{Review of estimates on Bessel functions on a compact set}
\noindent
In this Section, we will recall some uniform asymptotics for the Bessel function $J_{\nu}(r)$ and the Hankel functions $H_{\nu}^{(i)}(r)$ with respect to $r$ when $r$ 
belongs to a real compact set. We emphasize that these uniform asymptotics fail if $r \in (0,+\infty)$.

\begin{prop}[\cite{DN5}, Proposition A.7]\label{rappelestiB}
 Let $\delta > 0$ be small enough. For $r > 0$ belonging to a compact set, we have the following uniform asymptotics when $\nu \to \infty$
 \begin{equation}\label{asJnu}
  J_{\nu}(r) = \frac{1}{\Gamma(\nu+1)} \left( \frac{r}{2} \right)^{\nu} \left( 1 + O\left( \frac{1}{\nu} \right) \right), \quad |Arg(\nu)| \leq \pi - \delta,
 \end{equation}
  \begin{equation}\label{asHnu1}
  H_{\nu}^{(1)}(r) = -\frac{i}{\pi} \Gamma(\nu) \left( \frac{r}{2} \right)^{-\nu} \left( 1 + O\left( \frac{1}{\nu} \right) \right), \quad |Arg(\nu)| \leq \frac{\pi}{2} - \delta,
 \end{equation}
 and
 \begin{equation}\label{asHnu2}
  H_{\nu}^{(2)}(r) = \overline{H_{\overline{\nu}}^{(1)}(r)} = \frac{i}{\pi} \Gamma(\nu) \left( \frac{r}{2} \right)^{-\nu} \left( 1 + O\left( \frac{1}{\nu} \right) \right), \quad |Arg(\nu)| \leq \frac{\pi}{2} - \delta.
 \end{equation}
\end{prop}
\noindent
On the imaginary axis the asymptotics of the Hankel functions for a parameter $r$ belonging to a compact set are given by the following Proposition.

\begin{prop}\label{estiHiR}
 For $r > 0$ belonging to a compact set, the Hankel functions satisfy the following asymptotics when $\nu = iy$ and $|y| \to +\infty$
 \begin{equation}\label{estiH1iR}
    |H_{iy}^{(1)}(r)| = \frac{\sqrt{2}}{\sqrt{\pi |y|}} e^{\frac{\pi}{2} y} (1 + o(1)),
 \end{equation}
and
 \begin{equation}\label{estiH2iR}
    |H_{iy}^{(2)}(r)| = \frac{\sqrt{2}}{\sqrt{\pi |y|}} e^{-\frac{\pi}{2} y} (1 + o(1)).
 \end{equation}
 \end{prop}

\begin{proof}
 We recall that
  \[ H_{iy}^{(1)} (r) = \frac{e^{\pi y} J_{iy}(z)-J_{-iy}(z)}{\sinh(\pi y)}.\]
%
 First, thanks to (\ref{asJnu}),
  \[ J_{iy}(r) = \frac{1}{\Gamma(1+iy)} \left( \frac{r}{2} \right)^{iy} \left( 1 + O\left( \frac{1}{|y|} \right) \right), \quad |y| \to +\infty.\]
 So,
  \[ |J_{iy}(r)| = \frac{1}{|\Gamma(1+iy)|} \left( 1 + O\left( \frac{1}{|y|} \right) \right) = \frac{1}{|y| |\Gamma(iy)|} \left( 1 + O\left( \frac{1}{|y|} \right) \right), \quad |y| \to +\infty.\]
 We know (see \cite{Leb}, p.14) that for real $y$
 \[ |\Gamma(iy)|^2 = \frac{\pi}{y \sinh(\pi y)}.\]
 Thus,
  \[ |J_{iy}(r)| = \sqrt{\frac{|\sinh(\pi y)|}{|y| \pi}} \left( 1 + O\left( \frac{1}{|y|} \right) \right), \quad |y| \to +\infty.\]
  Moreover, using (\ref{lienBessel1}) we obtain that
  \[ |J_{-iy}(r)| = |\overline{J_{-iy}(r)}| = |J_{iy}(r)| = \sqrt{\frac{|\sinh(\pi y)|}{|y| \pi}} \left( 1 + O\left( \frac{1}{|y|} \right) \right), \quad |y| \to +\infty.\]
  We now should split our study into two cases.\\
  \\
 \noindent
 \underline{Case 1: $y \geq 0$.} In this case 
  \[ J_{-iy}(z) = o(e^{\pi y} J_{iy}(z)), \quad y \to +\infty.\]
  So, when $y \to + \infty$,
  \begin{eqnarray*}
     | H_{iy}^{(1)} (r) | &=& \left| \frac{e^{\pi y} J_{iy}(z)}{\sinh(\pi y)}\right|(1+o(1))  \\
     &=& \frac{1}{\sqrt{\sinh(\pi y) y \pi}} e^{\pi y} \left( 1 + o(1)\right)\\
     &=& \frac{\sqrt{2}}{\sqrt{\pi y}} e^{\frac{\pi}{2} y} (1 + o(1)),
  \end{eqnarray*}
where we used for the last step
\[ \sinh(\pi y) = \frac{e^{\pi y} - e^{-\pi y}}{2} \sim \frac{e^{\pi y}}{2}, \quad y \to + \infty.\]
\\
 \noindent
 \underline{Case 2: $y \leq0$.} In this case 
  \[  e^{\pi y} J_{iy}(z)= o(J_{-iy}(z)), \quad y \to -\infty.\]
  So, when $y \to - \infty$,
  \[ | H_{iy}^{(1)} (r) | = \frac{\sqrt{2}}{\sqrt{\pi |y|}} e^{\frac{\pi}{2} y} (1 + o(1)).\]
 Using these two asymptotics we then obtain (\ref{estiH1iR}).\\
 \noindent
  Finally, we prove (\ref{estiH2iR}) by using (\ref{lienBessel2}) as follows
  \[ |H_{iy}^{(2)}(r)| = |\overline{H_{iy}^{(2)}(r)}| = |H_{-iy}^{(1)}(r)| = \frac{\sqrt{2}}{\sqrt{\pi |y|}} e^{-\frac{\pi}{2} y} (1 + o(1)), \quad |y| \to + \infty.\]
\end{proof}

\subsection{Estimates on the Green kernels}\label{estigreen}
\noindent
The aim of this Section is to obtain useful estimates on several Green kernels. We recall that the Green kernel $N(r,s,\nu)$ is defined by
 \[ N(r,s,\nu) = u(r)v(s) - u(s)v(r),\]
where $(u,v)$ is a Fondamental System of Solutions of the free stationary equation
\[ -u'' +  \frac{(\nu_R)^2 -  \frac{1}{4}}{r^2}  u = u, \quad \forall r \geq r_0,\]
defined by
\[ u(r) = \sqrt{\frac{\pi r}{2}} J_{\nu_R}(r) \quad \text{and} \quad v(r) = -i \sqrt{\frac{\pi r}{2}} H_{\nu_R}^{(1)}(r),\]
where $\nu_R = \nu - \gamma(R)$.
The main idea to obtain an estimate on $N(r,s,\nu)$ on the half plane $\{\text{Re}(\nu) \geq \gamma(R) \}$ is to use the Phragm\'en-Lindel\"of Theorem because $u$ and $v$ (and so
$N(r,s,\nu)$) are entire functions of order $1$ with infinite type. We thus roughly speaking just need to estimate $N(r,s,\nu)$ for $\nu \geq \gamma(R)$ and for 
$\nu \in i \R+\gamma(R)$. We first study the behaviour of $N(r,s,\nu)$ when $\nu \to + \infty$. Thanks to Proposition \ref{rappelestiB}, we can first prove the following Lemma.

\begin{lemma}\label{estiprod}
 For $r,s \geq r_0$ belonging to a compact set, the following uniform asymptotics are satisfied when $\nu_R \to +\infty$
 \[ u(r)v(s) = \frac{-1}{2} \frac{\sqrt{rs}}{\nu_R} \left( \frac{r}{s} \right)^{\nu_R} \left( 1 + O \left( \frac{1}{\nu_R} \right) \right)\]
 and
 \[ u(s)v(r) = \frac{-1}{2} \frac{\sqrt{rs}}{\nu_R} \left( \frac{s}{r} \right)^{\nu_R} \left( 1 + O \left( \frac{1}{\nu_R} \right) \right).\]
\end{lemma}
%

\begin{prop}\label{estiNR}
  For $r_0 \leq r \leq s \leq R$ belonging to a compact set, we have the uniform asymptotics when $\nu \to +\infty$
 \[ N(r,s,\nu) = \frac{1}{2} \frac{\sqrt{rs}}{\nu_R} \left( \frac{s}{r} \right)^{\nu_R} \left( 1 + o(1) \right).\]
\end{prop}

\begin{proof}
 For $r \leq s$, we know that when $\nu \to + \infty$,
 \[ \left( \frac{r}{s} \right)^{\nu_R}  = o\left( \left( \frac{s}{r} \right)^{\nu_R} \right).\]
Therefore,
\[ u(r)v(s) = o(u(s)v(r))\]
 and then using that
 \[ N(r,s,\nu) = u(r)v(s) - u(s)v(r) = -u(s)v(r)(1 + o(1)), \quad \nu \to + \infty,\]
 we obtain the result we want thanks to Lemma \ref{estiprod}.
\end{proof}
\noindent
We now need to obtain an estimate of $N(r,s,\nu)$ for $\nu \in i\R+\gamma(R)$ which is given by the following Proposition.

\begin{prop}\label{estiNiR}
 There exists a positive constant $C$ such that for $r,s > 0$ belonging to a compact set, we have the following estimate when $\nu = iy+\gamma(R)$ and $|y| \to +\infty$
 \[ |N(r,s,iy+\gamma(R))| \leq \frac{C}{|y|}.\]
\end{prop}

\begin{proof}
 We recall that
 \begin{eqnarray*}
  N(r,s,\nu) &=& u(r)v(s)-u(s)v(r)\\
  &=& \frac{i \pi}{2} \sqrt{rs} \left( J_{\nu_R}(s) H_{\nu_R}^{(1)}(r) - J_{\nu_R}(r)H_{\nu_R}^{(1)}(s) \right)\\
  &=& \frac{i \pi}{4} \sqrt{rs} \left( H_{\nu_R}^{(1)}(r) H_{\nu_R}^{(2)}(s) - H_{\nu_R}^{(1)}(s) H_{\nu_R}^{(2)}(r) \right)
 \end{eqnarray*}
 and we use Proposition \ref{estiHiR} since $\nu_R = \nu - \gamma(R) = iy$.
\end{proof}

\begin{prop}\label{estiN1proof}
 There exists a constant $C > 0$ such that for all $\nu \in \{ \text{Re}(\nu) \geq \gamma(R) \}$ and for all $r_0 \leq r \leq s \leq R$
 \[ |N(r,s,\nu)| \leq \frac{C}{|\nu_R|+1} \left( \frac{s}{r} \right)^{\text{Re}(\nu_R)} \leq \frac{C}{|\nu_R|+1} \left( \frac{R}{r_0} \right)^{\text{Re}(\nu_R)}.\]
\end{prop}

\begin{proof}
 We recall that
 \[ N(r,s,\nu) = u(r)v(s) - u(s)v(r).\]
Since $u$ and $v$ are entire functions of order $1$ with infinite type, $N(r,s,\nu)$ is also an entire functions of order $1$ with infinite type and thanks to the 
Phragm\'en-Lindel\"of Theorem we thus just need to estimate $N(r,s,\nu)$ for $\nu \geq \gamma(R)$ and for $\nu \in i \R+\gamma(R)$.\\
\\
\noindent
\underline{Estimate of $N(r,s,\nu)$ on $i \R+\gamma(R)$.} We use Proposition \ref{estiNiR} which gives us that
\[ |N(r,s,iy+\gamma(R))| \leq \frac{C}{|y|}, \quad |y| \to +\infty.\]
Therefore, there exists a positive constant $C$ such that for all $y \in \R$,
\[ |N(r,s,iy+\gamma(R))| \leq \frac{C}{|y|+1}.\]
\\
\noindent
\underline{Estimate of $N(r,s,\nu)$ on $[\gamma(R),+\infty)$.} We use Proposition \ref{estiNR} which gives us that
 \[ |N(r,s,\nu)| \leq \frac{C}{\nu_R +1} \left( \frac{s}{r} \right)^{\nu_R}, \quad \nu \to +\infty.\]
Therefore, there exists a positive constant $C$ such that for all $\nu \geq \gamma(R)$,
 \[ |N(r,s,\nu)| \leq \frac{C}{\nu_R+1} \left( \frac{s}{r} \right)^{\nu_R}.\]
 \\
 \noindent
\underline{Conclusion.} We introduce a new function
\[f(\nu) = \left( \frac{r}{s} \right)^{\nu_R} (\nu_R + 1) N(r,s,\nu).\]
By definition, $f$ is a holomorphic function of order $1$ and with infinite type on $\{\text{Re}(\nu) \geq \gamma(R) \}$. Moreover, thanks to the previous estimates $f$ is bounded 
on $[\gamma(R),+\infty)$ and $i \R + \gamma(R)$. Therefore, thanks to the Phragm\'en-Lindel\"of Theorem we can conclude that $f$ is bounded on $\{\text{Re}(\nu) \geq \gamma(R) \}$. By definition, it provides
us that for all $ \nu \in \{\text{Re}(\nu) \geq \gamma(R) \}$
 \[ |N(r,s,\nu)| \leq \frac{C}{\nu_R+1} \left( \frac{s}{r} \right)^{\text{Re}(\nu_R)}.\]
\end{proof}

\begin{lemma}\label{estiK1proof}
 For all $\nu \geq \gamma(R) $ and for all $r_0 \leq r \leq s \leq R$
 \[ |K(r,s,\nu)| = \frac{s}{2 \nu_R} ( 1 + o(1)).\]
 where
 \[ K(r,s,\nu) = \frac{F_0^+(s,\nu)}{F_0^+(r,\nu)} N(r,s,\nu).\]
\end{lemma}

\begin{proof}
 First, we use Proposition \ref{estiNR} which provides us that when $\nu \to +\infty$
 \[ |N(r,s,\nu)| = \frac{1}{2} \frac{\sqrt{rs}}{\nu_R} \left( \frac{s}{r} \right)^{\nu_R} \left( 1 + o(1) \right).\]
 Secondly, we recall that
 \[ F_0^+(r,\nu) = e^{i\left( \nu_R + \frac{1}{2} \right) \frac{\pi}{2}} \sqrt{\frac{\pi r}{2}} H_{\nu_R}^{(1)}(r), \quad \forall r \geq r_0.\]
 Hence, when $\nu \to + \infty$, thanks to (\ref{asHnu1})
 \begin{eqnarray*}
   |F_0^+(r,\nu)| &=& \sqrt{\frac{\pi r}{2}} |H_{\nu_R}^{(1)}(r)| \\
   &=& \sqrt{\frac{r}{2\pi}} \Gamma(\nu_R) \left( \frac{r}{2} \right)^{-\nu_R} \left( 1 + O\left( \frac{1}{\nu_R} \right) \right).
 \end{eqnarray*}
Therefore,
\[ \frac{|F_0^+(s,\nu)|}{|F_0^+(r,\nu)|} = \left( \frac{r}{s} \right)^{\nu_R-\frac{1}{2}} \left( 1 + O\left( \frac{1}{\nu_R} \right) \right), \quad \nu \to + \infty.\]
Finally, we thus obtain when $\nu \to +\infty$
\begin{eqnarray*}
 |K(r,s,\nu)| &=& \frac{|F_0^+(s,\nu)|}{|F_0^+(r,\nu)|} |N(r,s,\nu)|\\
 &=& \frac{1}{2} \frac{\sqrt{rs}}{\nu_R}  \left( \frac{r}{s} \right)^{\nu_R-\frac{1}{2}} \left( \frac{s}{r} \right)^{\nu_R} \left( 1 + o(1) \right)\\
 &=& \frac{s}{2 \nu_R} ( 1 + o(1)).
\end{eqnarray*}
\end{proof}\vspace{0,5cm}

\noindent
\textit{Acknowledgments:} This paper was initiated by Thierry Daud\'e and Fran\c{c}ois Nicoleau during the PhD of the author.
The author wants to deeply thank Thierry Daud\'e and Fran\c{c}ois Nicoleau for their help and their support and Niky Kamran for his encouragement.

\newpage
{}


\end{document}